\documentclass[envcountsect,envcountsame]{llncs}


\newcommand{\short}[1]{}
\newcommand{\full}[1]{#1}

\usepackage{amsthm}
\usepackage{makeidx}
\usepackage{ifthen}
\usepackage{stmaryrd}
\newboolean{showproofs} 
\usepackage{wrapfig}

\usepackage[english]{babel}
\usepackage[cmex10]{amsmath}
\usepackage{amsfonts}
\usepackage{amssymb}
\usepackage{stmaryrd}
\usepackage{multirow}
\short{\usepackage{mathpartir}}
\usepackage{color}
\usepackage{graphicx}
\definecolor{dkblue}{rgb}{0,0.1,0.6}
\definecolor{dkgreen}{rgb}{0,0.35,0}
\definecolor{dkviolet}{rgb}{0.3,0,0.5}
\definecolor{dkred}{rgb}{0.5,0,0}
\usepackage{listings}
\usepackage{lstnt}

\setboolean{showproofs}{true} 

\usepackage{hyperref}
\usepackage[latin1]{inputenc}
\usepackage{enumerate}
\usepackage{newlfont}
\usepackage{bbold}
\usepackage{graphicx} 
\usepackage{microtype} 
\usepackage{placeins} 
\usepackage{tikz}
\usepackage{amsmath}
\usepackage{amssymb,latexsym}
\usepackage{mathrsfs}
\usepackage{amsthm}
\usepackage{times}
\usepackage{listings}
\usepackage[capitalize]{cleveref}
\usepackage{aliascnt}
\usepackage[colorinlistoftodos, textsize=footnotesize,color=orange!70]{todonotes}
\usepackage[arrow,matrix,frame,curve,cmtip]{xy}\CompileMatrices
\usepackage{algorithm}
\usepackage{algpseudocode}

\usetikzlibrary{shapes,arrows} 
\usetikzlibrary{positioning}
\usetikzlibrary{chains}
\usetikzlibrary{automata}
\usetikzlibrary{patterns}
\usetikzlibrary{decorations.pathreplacing}
\usetikzlibrary{calc,decorations.markings,decorations.pathreplacing,fit,backgrounds,shapes.symbols,shapes.geometric}
\tikzset{
  gnode/.style={draw,shape=circle,inner sep=0,minimum height=.2cm,
    minimum width=.2cm},
  hyperedge/.style={shape=rectangle,draw,inner sep=0,minimum width=.6cm,
    minimum height=.4cm}
}
 
\tikzset{every fit/.style={shape=rectangle,inner sep=5pt}}

\tikzset{
  mono/.style={>->},  
  ontop/.style={preaction={draw,-,line width=3pt,white}},
  arlab/.style={circle,inner sep=1pt,font=\scriptsize}
}

\renewcommand{\epsilon}{\varepsilon}
\newcommand{\str}[1]{\txt{\scriptsize #1}}

\newcommand{\precsimu}{\raisebox{-3pt}{$\,\stackrel{\normalsize
      \sqsubset}{\footnotesize \sim}\,$}}
\newcommand{\beq}{b_1}
\newcommand{\bin}{b_2}
\newcommand{\HKPA}{\texttt{HKP}_{\hspace{-0.1cm}\mathcal{A}}}
\newcommand{\HKA}{\texttt{ABK}}

\usepackage{hyperref}
\usepackage{ifthen}
\usepackage{rotating}

\renewcommand{\bigsqcap}{\rotatebox[origin=c]{180}{\ensuremath{\bigsqcup}}}

\renewcommand{\phi}{\varphi}
\renewcommand{\epsilon}{\varepsilon}

\usepackage{thmtools}
\declaretheorem[
    name=Theorem,
    numberwithin=section]{thm}
\Crefname{thm}{Theorem}{Theorems}
\declaretheorem[
    name=Lemma,
    sibling=thm]{lem}
\Crefname{lem}{Lemma}{Lemmas}
\declaretheorem[
    name=Definition,
    sibling=thm]{defi}
\Crefname{defi}{Definition}{Definitions}
\declaretheorem[
    name=Corollary,
    sibling=thm]{cor}
\Crefname{cor}{Corollary}{Corollaries}
\declaretheorem[
    name=Proposition,
    sibling=thm]{prop}
\Crefname{prop}{Proposition}{Propositions}

\Crefname{algo}{Algorithm}{Algorithms}

\Crefname{rem}{Remark}{Remarks}
\declaretheorem[
    name=Example,
    sibling=thm]{ex}
\Crefname{ex}{Example}{Examples}

\usepackage{enumitem}
\newlist{thmlist}{enumerate}{1}
\setlist[thmlist]{label=(\roman{thmlisti}), ref=\thethm.(\roman{thmlisti}),noitemsep}
\newlist{lemlist}{enumerate}{1}
\setlist[lemlist]{label=(\roman{lemlisti}), ref=\thethm.(\roman{lemlisti}),noitemsep}

\Crefname{listlem}{Lemma}{Lemmas}
\Crefname{listthm}{Theorem}{Theorems}
\Crefname{thmlisti}{Theorem}{Theorems}
\Crefname{lemlisti}{Lemma}{Lemmas}
\creflabelformat{listthm}{#2\thethm.#1#3}
\creflabelformat{listlem}{#2\thethm.#1#3}

\author{Filippo Bonchi\inst{1} \and Barbara König\inst{2} \and Sebastian Küpper\inst{2}}
\institute{ENS Lyon, France\and Universität Duisburg-Essen, Germany 
}\title{Up-To Techniques for Weighted Systems\full{\\ (Extended Version)}\thanks{Research
    partially supported by DFG project BEMEGA and ANR-16-CE25-0011 REPAS.}}

\ifthenelse{\boolean{showproofs}}{}{
	\usepackage{environ}
	\NewEnviron{killcontents}{}

}

\full{\renewcommand{\todo}[1]{}}

\begin{document}

\maketitle


\begin{abstract}
  We show how up-to techniques for (bi-)similarity can be used in the
  setting of weighted systems. The problems we consider are language
  equivalence, language inclusion and the threshold problem (also
  known as universality problem) for weighted automata.
  We build a bisimulation relation on the fly and work up-to
  congruence and up-to similarity. This requires to determine whether
  a pair of vectors (over a semiring) is in the congruence closure of
  a given relation of vectors. This problem is considered for rings
  and $l$-monoids, for the latter we provide a rewriting algorithm and
  show its confluence and termination.
  We then explain how to apply these up-to techniques to weighted
  automata and provide runtime results.
\end{abstract}

\section{Introduction}
\label{sec:intro}

Language equivalence of deterministic automata can be checked by means
of the bisimulation proof principle.  For non-deterministic automata,
this principle is sound but not complete: to use bisimulation, one
first has to determinize the automaton, via the so-called powerset
construction. Since the determinized automaton might be much larger
than the original non-deterministic one, several algorithms
\cite{CAV06,DoyenR10,achmv:simulation-meets-antichains,bp:checking-nfa-equiv}
have been proposed to perform the determization on the fly and to
avoid exploring a huge portion of states. Among these, the algorithm
in \cite{bp:checking-nfa-equiv} that exploits \emph{up-to techniques}
is particularly relevant for our work.

Up-to techniques have been introduced by Robin Milner in his seminal work on CCS \cite{Milner89} and, since then, they proved useful, if not essential, in numerous proofs about concurrent systems (see \cite{PS11} for a list of references). According to the standard definition a relation $R$ is a bisimulation whenever two states $x,y$ in $R$  can simulate each
other, resulting in a pair $x',y'$ that is still in $R$. An up-to
technique allows to replace the latter $R$ by a larger relation
$f(R)$ which contains more pairs and hence allows to cut off
bisimulation proofs and work with much smaller relations.

%
%
%
%

Here we focus on up-to techniques in a quantitative setting: weighted
systems, especially weighted automata over arbitrary semirings. Some
examples of up-to techniques for weighted systems already appeared
in \cite{BonchiPPR14} and \cite{rbbrps:coalgebraic-up-to}, that study up-to techniques
from the abstract perspective of coalgebras.

\medskip

Although up-to techniques for weighted systems have already received
some attentions, their relevance for algorithms to perform behavioural analysis has never
been studied properly.  This is the main aim of our paper: we give a
uniform class of algorithms exploiting up-to techniques to solve the
problems of equivalence, inclusion and universality, which, in the
weighted setting, asks whether the weight of all words is below some
given threshold. In particular we show how to implement these
techniques and we perform runtime experiments.

\medskip

The key ingredient to algorithmically exploit up-to techniques is a
procedure to decide, given $x,y,R$ as above, whether $x, y$ belongs to
$f(R)$.  For a non-deterministic automaton (NFA) with state space $S$, the
algorithm in \cite{bp:checking-nfa-equiv} uses as sub-routine a
rewriting system to check whether two sets of states
$S,S' \in \mathcal{P}(X)$ -- representing states of the determinised
automaton -- belong to $c(R)$, the congruence closure of $R$.

For NFA, the congruence closure is taken with
respect to the structure of join semi-lattices
$(\mathcal{P}(X),\cup, \emptyset)$, carried by the state space of a
determinized automaton. For weighted automata, rather than join
semi-lattices, we need to consider the congruence closure for
\emph{semimodules} (which resemble vector spaces, but are defined over
semirings instead of fields). Indeed, an analogon of the powerset
construction for weighted automata results in a sort of ``determinised
automaton'' (called in \cite{BonchiBBRS12} linear weighted automaton)
whose states 
are vectors with values in the underlying semiring.

Our first issue is to find a procedure to check whether two vectors belong to the congruence closure (with respect to semimodules) of a given relation. We
face this problem for different semirings, especially rings and
$l$-monoids. For $l$-monoids we adapt the rewriting procedure for the
non-deterministic case \cite{bp:checking-nfa-equiv} and show its confluence and
termination, which guarantees a unique normal form as a representative
for each equivalence class. Confluence holds in general and
termination can be shown for certain semirings, such as the tropical
semiring (also known as the $(\min,+)$-semiring).

Reasoning up-to congruence is sound for language equivalence, but not
for inclusion. For the latter, we need the precongruence closure that,
in the case of $l$-monoids, can be checked with a simple modification
of the rewriting procedure. Inspired by
\cite{achmv:simulation-meets-antichains}, we further combine this
technique with a certain notion of weighted \emph{similarity}, a
preorder that entails language inclusion and can be computed in
polynomial time.

We then show how to apply our up-to techniques to language equivalence
and inclusion checks for weighted automata. For some interesting
semirings, such as the tropical semiring, these problems are known to
be undecidable \cite{Krob94theequality}. But based on the inclusion
algorithm we can develop an algorithm which solves the universality
(also called threshold) problem for the tropical semiring over the
natural numbers.  This problem is known to be
$\mathsf{PSPACE}$-complete and we give detailed runtime results that
compare our up-to threshold algorithm with one previously introduced
in \cite{Alma11whatsdecidable}.

\section{Preliminaries}

In this section we recall all the algebraic structures we intend to
work with and, in particular, spaces of vectors over these structures.

A \emph{semiring} is a tuple $\mathbb S=(S,+, \cdot, 0, 1)$ where
$(S,+,0)$ is a commutative monoid, $(S,\cdot,1)$ is a monoid,
\emph{$0$ annihilates $\cdot$} (i.e., $0\cdot s_1=0 = s_1\cdot 0$) and
\emph{$\cdot$ distributes over $+$} (i.e.,
$(s_1+ s_2)\cdot s_3=s_1\cdot s_3+ s_2\cdot s_3$ and
$s_3\cdot (s_1+ s_2)=s_3\cdot s_1+ s_3\cdot s_2$ ).  A \emph{ring} is
a semiring equipped with inverses for $+$.

Let $(L,\sqsubseteq)$ be a partially ordered set.  If for all pairs of
elements $\ell_1,\ell_2\in L$ the infimum $\ell_1 \sqcap \ell_2$ and
the supremum $\ell_1 \sqcup \ell_2$ exist (wrt.\ the order
$\sqsubseteq$), it is a \emph{lattice}. If
$(\ell_1\sqcup \ell_2)\sqcap \ell_3=(\ell_1\sqcap \ell_3)\sqcup
(\ell_2\sqcap \ell_3)$
for all $\ell_1,\ell_2,\ell_3\in L$, it is called
\emph{distributive}. It is \emph{complete} if suprema and
infima of arbitrary subsets exist.  Every complete distributive
lattice is a semiring $(L,\sqcup,\sqcap,\bot,\top)$, where $\bot$,
$\top$ are the infimum and supremum of $L$.

Let $(L,\sqsubseteq)$ be a lattice and $(L,\cdot,1)$ be a monoid.  If
$\cdot$ distributes over $\sqcup$, we call $(L,\sqcup, \cdot)$ an
$l$-monoid.  Moreover, if $L$ has a $\bot$-element $0$ that
annihilates $\cdot$, we call $(L,\sqcup,\cdot)$ bounded. and it is
then a semiring $(L,\sqcup,\cdot,0,1)$
It is called \emph{completely distributive} if $(L,\sqsubseteq)$ is
complete and multiplication distributes over arbitrary suprema.
Observe that every completely distributive $l$-monoid
is 
bounded.\footnote{Completely distributive $l$-monoids are often
  referred to as \emph{unital quantales}.}  It is called
\emph{integral} 
if $\top=1$.

\begin{ex}
  \label{ex:lmonoid}
  The tropical semiring is the structure
  $\mathbb T = (\mathbb R_0^+\cup\{\infty\}, \min, +,\infty,
  0)$.\footnote{We will sometimes use $\min$ as an infix operator (i.e.,
    $a\min\ b$).}  $\mathbb T$ is a distributive $l$-monoid for the
  lattice $(\mathbb R_0^+\cup\{\infty\}, \ge)$.  

  Another example for a distributive $l$-monoid is
  $\mathbb M = ([0,1], \max, \cdot, 0,1)$, which is based on the
  lattice $([0,1],\le)$.


  The $l$-monoid $\mathbb M$ is isomorphic to $\mathbb T$ via the
  isomorphism $\phi\colon \mathbb T \to \mathbb M, x\mapsto 2^{-x}$.
\end{ex}

Hereafter, we will sometimes identify the semiring $\mathbb S$ with
the underlying set $S$. For the sake of readability, we will only
consider commutative semirings
, i.e., semirings
where multiplication is commutative.

For a semiring $\mathbb S$ and a finite set $X$, an
\emph{$\mathbb S$-vector of dimension $X$} is a mapping
$v\colon X\rightarrow \mathbb S$. The set of all such vectors is
denoted by ${\mathbb S}^X$ and is called a \emph{semimodule}.

For notational convenience, we assume that $X=\{1,2,\ldots, |X|\}$ and
we write a vector $v$ as a column vector.
For $X$ and $Y$ finite sets, an \emph{$\mathbb S$-matrix of dimension
  $X\times Y$} is a mapping $M \colon X\times Y\rightarrow\mathbb S$.
The set of all such matrices is denoted by ${\mathbb S}^{X\times Y}$.
$M[x,y]$ ($v[x]$) denotes the $(x,y)$-th entry of $M$ ($x$-th entry of $v$). Furthermore $v\cdot s$
denotes the multiplication of a vector with a scalar $s$ and $v_1+v_2$
is the componentwise addition.
  
Given a set $V$ of $\mathbb S$-vectors, a \emph{linear combination} of
vectors in $V$ is a vector $v_1\cdot s_1 + \dots + v_n\cdot s_n$,
where $v_1,\dots, v_n\in V$, $s_1,\dots,s_n\in \mathbb S$. A subset of
$\mathbb{S}^X$ that is closed under linear combinations is called a
\emph{(sub-)semimodule}.

Henceforward we will always require $l$-monoids to be completely
distributive: this ensures that we have a residuation operation
defined as follows.

\begin{defi}
  The \emph{residuation operation} for a completely distributive
  $l$-monoid $\mathbb L$ is defined for all
  $\ell_1,\ell_2\in\mathbb L$ as 
  $\ell_1\rightarrow \ell_2=\bigsqcup\{\ell\in\mathbb L\mid
  \ell_1\cdot \ell\sqsubseteq \ell_2\}$,
  also called \emph{residuum} of $\ell_1,\ell_2$.  We extend this to
	$\mathbb L$-vectors, replacing $\ell_1,\ell_2$ by
	$v_1,v_2\in {\mathbb L}^X$.
\end{defi}


\begin{ex} Recall 
  $\mathbb T$, $\mathbb M$ in \cref{ex:lmonoid}.  For
  $\ell_1,\ell_2\in\mathbb{T}$ we have
  $\ell_1\rightarrow \ell_2 = \min\{\ell\in \mathbb{R}_0^+\cup
  \{\infty\}\mid \ell_1 + \ell\geq \ell_2\} = \ell_2 \,\dot{-}
  \,\ell_1$
  (modified subtraction). For $\ell_1, \ell_2 \in \mathbb{M}$, we have
  $\ell_1\rightarrow \ell_2=\max\{\ell \in [0,1]\mid \ell_1\cdot
  \ell\leq \ell_2\}=\min\{1, \frac{\ell_2}{\ell_1}\}$.
	
  Another example where the residuation operation can be easily
  characterized is any boolean algebra $(\mathbb B, \vee,\wedge,0,1)$.
  For $\ell_1, \ell_2\in\mathbb B$ we have
  $\ell_1\rightarrow \ell_2=\neg \ell_1\vee \ell_2$.
\end{ex}

We will assume that all relevant operations of any semiring under
consideration (addition, multiplication, in the case of $l$-monoids
residuation) are computable.

\section{Congruence Closure}
\label{sec:cong-closure}

As explained in the introduction, the key ingredient for exploiting
up-to techniques in Section \ref{sec:applications} is an algorithmic
procedure to check whether two vectors belong to the congruence
closure of a given relation of vectors.


\subsection{Problem Statement}
\label{sec:cong-closure-problem}

Let $X$ be a finite set and let $\mathbb S$ be a semiring.  A
  relation $R\subseteq \mathbb{S}^X\times \mathbb{S}^X$ is a
  \emph{congruence} if it is an equivalence and \emph{closed under linear
  combinations}, that is, for each $(v_1,v'_1), (v_2,v'_2)\subseteq R$ and each
  scalar $s\in\mathbb S$, 
  $(v_1+v_2, v_1'+v_2') \in R$ and $(v_1\cdot s,v'_1 \cdot s) \in R$.
  The \emph{congruence closure} $c(R)$ of a relation $R$ over a semiring
  $\mathbb S$ is the smallest congruence
  $R'\subseteq \mathbb{S}^X\times \mathbb{S}^X$ such that
  $R\subseteq R'$.  Alternatively, two vectors $v,v'\in\mathbb{S}^X$ are in
  $c(R)$ whenever this can be derived via the rules in
  \cref{tab:proof-rules}.

\begin{table}
\vspace{-0.3cm}
\begin{center}
  \begin{tabular}{c}
    (\textsc{Rel})\quad 
    $\begin{array}{c}
      v\ R\ w\\ \hline
      v\ c(R)\ w
    \end{array}$\quad\quad
    (\textsc{Refl})\quad 
    $\begin{array}{c}
      \\ \hline
      v\ c(R)\ v
    \end{array}$\quad\quad
    (\textsc{Sym})\quad 
    $\begin{array}{c}
      v\ c(R)\ w \\ \hline
      w\ c(R)\ v
    \end{array}$ \\[0.4cm] 
    (\textsc{Trans})\quad 
    $\begin{array}{c}
       u\ c(R)\ v\quad v\ c(R)\ w\\ \hline
       u\ c(R)\ w
    \end{array}$ \quad\quad
    (\textsc{Sca})\quad 
    $\begin{array}{c}
      v\ c(R)\ w \\ \hline
      v\cdot s\ c(R)\ w\cdot s
    \end{array}$ \quad \mbox{where $s\in\mathbb{S}$} \\[0.4cm] 
    (\textsc{Plus})\quad 
    $\begin{array}{c}
      v_1\ c(R)\ v'_1 \quad v_2\ c(R)\ v'_2\\ \hline
      v_1+v_2\ c(R)\ v'_1+v'_2
    \end{array}$ 
  \end{tabular}
\end{center}
\caption{Proof rules for the congruence closure}
\label{tab:proof-rules}
\vspace{-0.7cm}
\end{table}

Given a finite 
$R\subseteq \mathbb{S}^X\times \mathbb{S}^X$ and 
$v,w\in \mathbb{S}^X$, we aim to determine if
$(v,w)\in c(R)$.  


In \cite{bp:checking-nfa-equiv}, Bonchi and Pous presented a procedure
to compute the congruence closure for the two-valued boolean semiring
$B=\{0,1\}$. The purpose of this section is to generalise the procedure
towards more general semirings, such as rings and $l$-monoids.

\subsection{Congruence Closure for Rings}
\label{sec:cong-closure-rings}


A simple case to start our analysis is the congruence closure of a
ring.  It is kind of folklore (see
e.g. \cite{stark2003behaviour,b:weighted-bisimulation}) that a
submodules\footnote{A sub-semimodule for a ring is called submodule.}
can be used to represent a congruences. In particular we write $[V]$
to denote the submodule generated by a set of vectors $V$.



\begin{prop}\label{prop:congclosurerings}
  Let $\mathbb{I}$ be a ring and $X$ be a finite set.
  Let $R\subseteq \mathbb{I}^X\times \mathbb{I}^X$ be a relation and
  let $(v,v')\in \mathbb{I}^X\times \mathbb{I}^X$ be a pair of
  vectors.  We construct a generating set for a submodule of
  $\mathbb{I}^X$ by defining $U_R = \{u-u'\mid (u,u')\in R\}$.  Then
  $(v,v')\in c(R)$ iff $v-v'\in [U_R]$.

\end{prop}

This 
yields a
n algorithm for a congruence
check whenever we have a
n algorithm to solve linear equations,
e.g. for fields.
If the ring is not a field, it might still be possible to embed it
into a field.
In this case we can 
solve e.g. the language equivalence
problem (\cref{sec:langequiv}) for weighted automata in the field and
the results are also valid in the ring. Similarly, the procedure can
be used for probabilistic automata which can be seen as weighted
automata over the reals.
	




\subsection{Congruence Closure for $l$-Monoids}
\label{sec:cong-closure-l-monoids}

\subsubsection{Rewriting and Normal Forms.}


Our method 
to determine if a pair of vectors is in the
congruence closure is to employ a rewriting algorithm that rewrites both vectors to a normal form. The
se coincide iff the
vectors are related by the congruence closure.

%
%
%
%
%

\begin{defi}[Rewriting and normal forms]\label{def:RewritingNF}
  Let $\mathbb L$ be an integral $l$-monoid and let
  $R\subseteq {\mathbb L}^X\times {\mathbb L}^X$ be a finite relation.

  We define a set of rewriting rules $\mathcal R$ as follows: For each
  pair of vectors $(v,v')\in R$, we obtain two rewriting rules
  $v\mapsto v\sqcup v'$ and $v'\mapsto v\sqcup v'$.
	
  A rewriting step works as follows: given a vector $v$ and a
  rewriting rule $l\mapsto r$, we compute the residuum
  $l\rightarrow v$ and, provided
  $v\sqsubset (v\sqcup r\cdot(l\rightarrow v))$, the rewriting rule is
  applicable and $v$ rewrites to $v\sqcup r\cdot(l\rightarrow v)$
  (symbolically: $v\leadsto v\sqcup r\cdot(l\rightarrow v)$). A vector
  $v$ is in normal form wrt.\ $R$, provided there exists no rule that
  is applicable to $v$.
\end{defi}

\begin{ex}
  In order to illustrate how rewriting works, we work in $\mathbb T$,
  set $X=\{1,2\}$ (two dimensions) and take the
  relation
  $R = \{\left({\infty\choose
      0},{0\choose\infty}\right)\}\subseteq\mathbb T^2\times\mathbb
  T^2$,
  relating the two unit vectors, and the vector
  $v = {\infty\choose 3}$. This yields a rule
  $l={\infty\choose 0}\mapsto r={0\choose 0}$. We obtain $l\to v= 3$
  and hence
  $v\leadsto v\sqcup r\cdot(l\rightarrow v) = {\infty\choose 3} \min
  \left({0\choose 0}+3\right) = {3\choose 3}$.
\end{ex}

It is worth to observe that when $\mathbb L$ is the boolean semiring, the above procedure coincides with the one in \cite{bp:checking-nfa-equiv}.
The rewriting relation satisfies some simple properties:




\begin{lem}
  \label{lem:rewriting}
  \begin{lemlist}
  \item If $v\leadsto v'$ and $v\sqsubseteq w$, then $v'\sqsubseteq w$ or there exists
    $w'$ s.t. $w\leadsto w'$ and $v'\sqsubseteq w'$.
    \label{lem:rewriting-2}
  \item Whenever $v\leadsto v'$ and $w$ is any vector, there exists a vector $u$ s.t. $v\sqcup w \leadsto u \sqsupseteq v'\sqcup w$ or $v\sqcup w = v'\sqcup w$.
    \label{lem:rewriting-3}
  \end{lemlist}
\end{lem}

We now have to prove the following three statements: (i) Our technique
is sound, i.e. whenever two vectors have the same normal form wrt.\
$R$, they are in $c(R)$. (ii) Our technique is complete, i.e. whenever
two vectors are in $c(R)$, they have the same normal form wrt.\ $R$.
(iii) Our algorithm to compute normal forms terminates.

We will show (i) and prove that (ii) follows from (iii). Afterwards we
will discuss sufficient conditions and examples where (iii) holds.

\begin{thm}
  Whenever there exists a vector $v$, such that two vectors
  $v_1$,$v_2$ both rewrite to $\overline v$, i.e., $v_1\leadsto^* \overline v$,
  $v_2\leadsto^* \overline v$, then $(v_1,v_2)\in c(R)$.
  \label{theorem-correctness}
\end{thm}

\begin{proof}
  We will show that if $v$ rewrites to $v'$ via a rule $l\mapsto r$,
  then $(v,v')\in c(R)$. 

  Since $l\mapsto r$ is a rewriting rule we have that $l = w$,
  $r = w\sqcup w'$ for $(w,w')\in R$ or $(w',w)\in R$. In both cases
  $w = w\sqcup w\ c(R)\ w\sqcup w'$ due to the definition of
  congruence closure, using rules (\textsc{Plus}), (\textsc{Rel}) and
  (\textsc{Refl}), as well as (\textsc{Sym}) in case $(w',w)\in
  R$. Hence $l\ c(R)\ r$.
  This implies that $l\cdot (l\to v)\ c(R)\ r\cdot (l\to v)$
  (\textsc{Sca}) and furthermore
  $v\sqcup l\cdot (l\to v)\ c(R)\ v\sqcup r\cdot (l\to v)$
  (\textsc{Plus}). Since $l\cdot (l\to v) \sqsubseteq v$
		we have
  $v\sqcup l\cdot (l\to v) = v$ and hence $v\ c(R)\ v'$.
\end{proof}

This concludes the proof of soundness, we will go on proving
completeness.

\begin{lem}
  Assume we have a rewriting system that always terminates.
  Then the local Church-Rosser property holds. That is whenever
  $v\leadsto v_1$ and $v\leadsto v_2$, there exists a vector $v'$
  such that $v_1\leadsto^* v'$ and $v_2\leadsto^* v'$.
  \label{lem:ChurchRosser}
\end{lem}

If a rewriting system terminates and the local Church-Rosser property
holds, the system is automatically confluent \cite{dj:rewrite}. In
this case, every vector $v$ is as associated with a unique normal
form, written $\Downarrow_{\mathcal{R}}\! v$ or simply
$\Downarrow\! v$ where
$v\leadsto^* \Downarrow\! v \not\leadsto$.

Furthermore, due to \cref{lem:rewriting-2} we know that $\Downarrow$
is monotone, i.e., $v\sqsubseteq v'$ implies
$\Downarrow\! v \sqsubseteq \Downarrow\! v'$.  This also implies
$\Downarrow\! (v\sqcup v') \sqsupseteq (\Downarrow\! v)\sqcup
(\Downarrow\! v')$.

\begin{lem}
  \label{lem:prepcorrectness}
  For all $v\in\mathbb L^X, \ell\in\mathbb L$ we have that if
  $v \leadsto v'$, then $v\cdot \ell \leadsto v''$ for some
  $v'' \sqsupseteq v'\cdot \ell$ or $v\cdot\ell = v'\cdot \ell$.
  In particular, if rewriting terminates, we have $(\Downarrow\!
  v)\cdot \ell \sqsubseteq \,\Downarrow\!(v\cdot \ell)$.
\end{lem}

Now we have all the necessary ingredients to show that the technique
is complete, provided the computation of a normal form terminates.

\begin{thm}
  Assume that rewriting terminates.  If $v \ c(R)\ v'$ then
  $\Downarrow\! v=\Downarrow\! v'$.
  \label{thm:RewritingCorrectness}
\end{thm}

\subsubsection{Termination.}

One technique to prove termination is given in \cref{cor:terminationgeneral}: it suffices to show that the supremum of all
the elements reachable via $\leadsto$ is included in the congruence class. First we need the following result.

\begin{prop}
  If $v \ c(R)\ \overline v$, then $v\leadsto^* v'$ where
  $v'\sqsupseteq v\sqcup\overline v$. 
  \label{prop:rewritingtogreater}
\end{prop}

Now take $\overline v = \bigsqcup \{\hat{v} \mid v \leadsto^* \hat{v}
\}$. By the above proposition if $v \ c(R)\ \overline v$, then $v' =
\overline v $ and $v \leadsto^* \overline v $. Since $\leadsto$ is
irreflexive, $\overline v \not \leadsto$. If we assume that rule
application is fair, we can guarantee that $\overline v$ is eventually
reached in every rewriting sequence.

\begin{cor}
  \label{cor:terminationgeneral}
If $v \ c(R)\ \bigsqcup \{\hat{v} \mid v \leadsto^* \hat{v} \}$, then
the rewriting algorithm terminates, assuming that every rule that
remains applicable is eventually applied.
\end{cor}

%
%
%
%

\paragraph{Termination for Specific $l$-Monoids.}

We now study the $l$-monoid
$\mathbb{M} = ([0,1],\max,\cdot, 0,1)$ from \cref{ex:lmonoid} and show
that the rewriting algorithm terminates for this $l$-monoid. For the
proof we mainly use the pigeon-hole principle and exploit the
total ordering of the underlying lattice. Since $\mathbb{M}$ is
isomorphic to $\mathbb{T}$, we obtain termination for the tropical
semiring as a corollary.

\begin{thm}
  The rewriting algorithm terminates for the $l$-monoids
  $\mathbb{M}$ and $\mathbb{T}$.\label{thm:rewritingterminates01maxtimes}
\end{thm}

These results provide an effective procedure for checking congruence
closure over the tropical semiring.  We will mainly apply them to
weighted automata, but expect that they can be useful to solve other
problems. \short{For instance, in \cite{bkk:up-to-weighted-arxiv}, we
  show an interesting connection to the shortest path problem.}
\full{For instance, in Appendix~\ref{sec:dijkstra}, we show an
  interesting connection to the shortest path problem.}

\paragraph{Termination for Lattices.}

We next turn to lattices and give a sufficient condition for
termination on lattices. Obviously, rewriting terminates for lattices
for which the ascending chain condition holds (i.e., every ascending
chain eventually becomes stationary), but one can go beyond that.

In this section,  we assume a completely distributive
lattice $\mathbb L$ and a boolean algebra $\mathbb B$ such that the
orders of $\mathbb L$ and $\mathbb B$, as well as the infima
coincides. Suprema need not coincide.  
Thus, whenever there is ambiguity, we will add the index $\mathbb{B}$
or $\mathbb{L}$ to the operator. For the negation of a given
$x\in\mathbb B$, we write $\neg x$.
One way to obtain such a boolean algebra -- in particular one where
the suprema coincide as well -- is via Funayama's theorem, see
\cite{BGJ13}. \full{This embedding is also discussed in
  \cref{sec:appendix-congruence-closure}.}


We want to show that if $\mathbb L$ approximates $\mathbb B$ ``well
enough'', the rewriting algorithm terminates for $\mathbb L$.



\begin{thm}
  The approximation of an element $\ell\in\mathbb B$ in the lattice
  $\mathbb L$ is defined
  as 
  $\lfloor \ell\rfloor = {\bigsqcup}_{\mathbb L}\{\ell'\in\mathbb
  L\mid \ell'\leq \ell\}$.

  Let $\mathcal R$ be a rewriting system for vectors in $\mathbb L^X$.
  Whenever the set
  $L(l,x)=\{\ell\in\mathbb L\mid \lfloor \neg {l[x]\rfloor}\sqsubseteq
  \ell\sqsubseteq\neg {l[x]}\}$
  is finite for all rules $(l\mapsto r)\in\mathcal R$ and 
  all $x\in X$, rewriting terminates.
  \label{thm-termcond}
\end{thm}

Note that
$[\neg\ell] = [\ell\rightarrow_\mathbb{B} 0] =
\ell\rightarrow_\mathbb{L} 0$\full{
  (Lemma~\ref{lemslatticebool})}.  Hence the theorem says that there
should be only finitely many elements between the negation of an
element in the lattice and the negation of the same element in the
boolean algebra.
As a simple corollary we obtain that the rewriting algorithm
terminates for all boolean algebras.

\section{Up-To Techniques for Weighted Automata}
\label{sec:applications}


In this section we present applications of our congruence closure
method. More specifically, we investigate weighted automata and
present up-to techniques both for the language equivalence and the
inclusion problem, which are variants of the efficient up-to based
algorithm presented in \cite{bp:checking-nfa-equiv}. For the tropical
semiring we also give a procedure for solving the threshold problem,
based on the language inclusion algorithm.




\subsection{Language Equivalence for Weighted Automata}\label{sec:langequiv}

We turn our attention towards weighted automata and their languages.

A \emph{weighted automaton} over the semiring $\mathbb S$ and alphabet
$A$ is a triple $(X,o,t)$ where $X$ is a finite set of states,
$t= (t_a \colon X\rightarrow \mathbb S^X)_{a\in A}$ is an $A$-indexed
set of transition functions and $o \colon X\rightarrow\mathbb S$ is
the output function.  Intuitively $t_a(x)(y)=s$ means that the states
$x$ can make a transition to $y$ with letter $a\in A$ and weight
$s\in \mathbb S$ (sometimes written as $x\xrightarrow{a,s}y$).  The
functions $t_a$ can be represented as $X\times X$-matrices with values
in $\mathbb S$ and $o$ as a row vector of dimension $X$.  Given a
vector $v\in \mathbb S^X$, we use $t_a(v)$ to denote the vector
obtained by multiplying the matrix $t_a$ by $v$ and $o(v)$ to denote
the scalar in $\mathbb S$ obtained by multiplying the row vector $o$
by the column vector $v$.

	

A \emph{weighted language} is a function $\varphi\colon A^*\to \mathbb S$, where $A^*$ is the set of all words over $A$. We will use $\epsilon$ to denote the empty word and $aw$ the concatenation of a letter $a\in A$ with the word $w\in A^*$.  Every weighted automaton is associated with a function $\llbracket - \rrbracket \colon  \mathbb S^X \to  \mathbb S^{A^*}$ mapping  each vector into its \emph{accepted language}. For all $v\in \mathbb S^X$, $a\in A$ and $w\in A^*$, this is defined as
\[
\llbracket v \rrbracket (\epsilon) = o(v) \quad \qquad
\llbracket v \rrbracket (aw) = \llbracket t_a(v) \rrbracket (w)\text{.}
\]

Two vectors $v_1,v_2\in \mathbb S^X$ are called \emph{language
  equivalent}, written $v_1\sim v_2$ iff
$\llbracket v_1 \rrbracket = \llbracket v_2 \rrbracket$. \footnote{The
  accepted notions of language and language equivalence can be given
  for states rather than for vectors by assigning to each state
  $x\in X$ the corresponding unit vector $e_x\in \mathbb S^X$. On the
  other hand, when weighted automata are given with an initial vector
  $i$ -- which is often the case in literature -- one can define the
  language of an automaton as $\llbracket i \rrbracket$.}
The problem of checking language equivalence in weighted automata for
an arbitrary semiring is undecidable: for the tropical semiring this
was shown by Krob in \cite{Krob94theequality}; the proof was later simplified in \cite{Alma11whatsdecidable}.
However, for several semirings the problem is decidable, for instance
for all (complete and distributive) lattices. For finite non-deterministic automata, i.e., automata weighted over
the boolean semiring, Bonchi and Pous introduced in \cite{bp:checking-nfa-equiv} the algorithm \texttt{HKC}.
The name stems from the fact that the algorithm extends the procedure of
Hopcroft and Karp \cite{hk:equ-finite-automata} with congruence closure.

\cref{fig:hkc} shows the extension of \texttt{HKC} to
weighted automata over an arbitrary semiring: the code is the same as the one in \cite{bp:checking-nfa-equiv}, apart from the fact that, rather than exploring sets of states, the algorithm works with vectors in $\mathbb S^X$. The check at step \texttt{(3.2)} can be performed with the procedures discussed in Section \ref{sec:cong-closure}.

Below we prove that the algorithm is sound and complete, but termination can fail
in two ways: either the check at step \texttt{(3.2)} does not terminate, or the while loop at step \texttt{(3)} does not. 
For the tropical semiring we have seen that the check at step \texttt{(3.2)} can be effectively performed by rewriting (Theorem \ref{thm:rewritingterminates01maxtimes}). Therefore, due to Krob's undecidability result,  the while  loop at step \texttt{(3)} 
may not terminate.  For (distributive)
lattices, we have shown termination of rewriting under some additional
constraints (Theorem \ref{thm-termcond}); moreover the loop at \texttt{(3)} will always terminate,
because from a given finite set of lattice elements only finitely many
lattice elements can be constructed using infimum and supremum
\cite{KK16}.

%
%
%
%
%
%
%

\begin{figure}[t]
\centering
\underline{\texttt{HKC} $(v_1,v_2)$}
\begin{codeNT}
(1) $R := \emptyset$; $todo := \emptyset$
(2) insert $(v_1,v_2)$ into $todo$
(3) while $todo$ is not empty do 
   (3.1)  extract $(v_1',v_2')$ from $todo$
   (3.2)  if $(v_1',v_2')\in c(R)$ then continue
   (3.3)  if $o(v_1')\neq o(v_2')$ then return false
   (3.4)  for all $a\in A$, 
             insert $(t_a(v_1'),\,t_a(v_2'))$ into $todo$
   (3.5)  insert $(v_1',v_2')$ into $R$ 
(4) return true
\end{codeNT}
\caption{Algorithm to check the equivalence of vectors $v_1,v_2\in \mathbb{S}^X$ for a weighted automata $(X,o,t)$.}
\label{fig:hkc}
\end{figure}

%
%
%
%


To prove soundness of \texttt{HKC}, we introduce the notions of simulation and
bisimulation up-to.  Let $Rel_{\mathbb S^X}$ be the complete lattice
of relations over $\mathbb S^X$ and
$\beq \colon Rel_{\mathbb S^X} \to Rel_{\mathbb S^X} $ be the monotone
map defined for all $R \subseteq \mathbb S^X \times \mathbb S^X$ as
\begin{equation*}
\beq(R)= \{(v_1,v_2) \mid o(v_1)=o(v_2) \text{ and for all }a\in A, \; (t_a(v_1), t_a(v_2))\in R \} 
\end{equation*}

\begin{defi} 
  A relation $R \subseteq \mathbb S^X \times \mathbb S^X$ is a
  \emph{$\beq$-simulation} if $R\subseteq \beq(R)$, i.e., for all
  $(v_1,v_2)\in R$: (i) $o(v_1)=o(v_2)$; (ii) for all $a\in A$,
  $(t_a(v_1), t_a(v_2))\in R$.

  For a monotone map
  $f\colon Rel_{\mathbb S^X} \to Rel_{\mathbb S^X}$, a
  \emph{$\beq$-simulation up-to $f$} is a relation $R$ such that
  $R\subseteq \beq(f(R))$.
\end{defi}

It is easy to show (see e.g. \cite{PS11}) that $\beq$-simulation
provides a sound and complete proof technique for $\sim$. On the other
hand, not all functions $f$ can be used as sound up-to techniques.
\texttt{HKC} exploits the monotone function
$c\colon Rel_{\mathbb S^X} \to Rel_{\mathbb S^X}$ mapping each
relation $R$ to its congruence closure $c(R)$.

\begin{prop}
  \label{prop:coinduction} Let $v_1,v_2\in \mathbb S^X$. It holds that
  $v_1\sim v_2$ iff there exists a $\beq$-simulation $R$ such that
  $(v_1,v_2)\in R$ iff there exists a $\beq$-simulation up-to $c$ $R$
  such that $(v_1,v_2)\in R$.
\end{prop}

With this result, it is easy to prove the correctness of the algorithm.

\begin{thm} \label{thm:langequicong} Whenever \texttt{HKC} terminates,
  it returns true iff
  $\llbracket v_1\rrbracket = \llbracket v_2\rrbracket$.
\end{thm}

\begin{proof}  
  Observe that $R\subseteq \beq (c(R)\cup todo)$ is an invariant for
  the while loop at step (3).

  If \texttt{HKC} returns $true$ then $todo$ is empty and thus
  $R\subseteq \beq (c(R))$, i.e., $R$ is a $\beq$-simulation up-to
  $c$.  By \cref{prop:coinduction}, $v_1\sim v_2$.
  
  Whenever \texttt{HKC} returns false, it encounters a pair
  $(v_1',v_2')\in todo$ such that $o(v_1')\neq o(v_2')$. Observe that
  for all pairs $(v_1',v_2')\in todo$, there exists a word
  $w=a_1a_2\dots a_n \in A^*$ such that
  $v_1'=t_{a_n}(\dots t_{a_2}(t_{a_1}(v_1)))$ and
  $v_2'=t_{a_n}(\dots t_{a_2}(t_{a_1}(v_2)))$.  Therefore
  $\llbracket v_1\rrbracket (w) = \llbracket v_1'\rrbracket (\epsilon)
  = o(v_1') \neq o(v_2') = \llbracket v_2'\rrbracket (\epsilon) =
  \llbracket v_2'\rrbracket (w)$.
\end{proof}



\subsection{Language Inclusion}\label{sec:langincl}
Whenever a semiring $\mathbb{S}$ carries a partial order $\sqsubseteq$, one
can be interested in checking language inclusion of the states of a
weighted automata $(X,o,t)$. More generally, given
$v_1,v_2\in \mathbb S^X$, we say that the language of $v_1$ is
included in the language of $v_2$ (written $v_1 \precsimu v_2$ ) iff
$\llbracket v_1\rrbracket(w)\sqsubseteq\llbracket v_2\rrbracket(w)$ for all
$w\in A^*$.

The algorithm \texttt{HKC} can be slightly modified to check language
inclusion, resulting in algorithm \texttt{HKP}: steps (3.2) and (3.3)
are replaced by
\begin{codeNT}
 (3.2)  if $(v_1',v_2')\in p(R)$ then continue
 (3.3)  if $o(v_1')\not\sqsubseteq o(v_2')$ then return false
\end{codeNT}
where $p\colon Rel_{\mathbb S^X} \to Rel_{\mathbb S^X}$ is the
monotone function assigning to each relation $R$ its pre-congruence
closure $p(R)$. 

\begin{wrapfigure}{R}{0pt}
  $(\textsc{Ord})\
  \begin{array}{c}
    v\sqsubseteq v'\\ \hline
    v\ p(R)\ v'
  \end{array}$
\end{wrapfigure}
The precongruence closure is defined as the closure of $R$ under $\sqsubseteq$,
transitivity and linear combination. That is, in the rules of
\cref{tab:proof-rules} $c(R)$ is replaced by $p(R)$,
rule~(\textsc{Sym}) is removed and rule~(\textsc{Refl}) is replaced by
rule~(\textsc{Ord}) on the right.

The soundness of the modified algorithm can be proved in the same way
as for \texttt{HKC} by replacing $c$ by $p$ and $\beq$ by
$\bin \colon Rel_{\mathbb S^X} \to Rel_{\mathbb S^X} $ defined for all
$R \subseteq \mathbb S^X \times \mathbb S^X$ as
\begin{equation*}
  \bin(R)= \{(v_1,v_2) \mid o(v_1) \sqsubseteq o(v_2) \text{ and for all }
  a\in A, \; (t_a(v_1), t_a(v_1))\in R \} \text{.}
\end{equation*}
However, the soundness of up-to reasoning is guaranteed only if
$\sqsubseteq$ is a precongruence, that is $p(\sqsubseteq)$ is contained in $\sqsubseteq$.

\begin{thm} 
  \label{thm:langeincl} 
  Let $\mathbb{S}$ be a semiring equipped with a precongruence
  $\sqsubseteq$.  Whenever \linebreak \texttt{HKP}($v_1,v_2$) terminates, it
  returns true if and only if $v_1\precsimu v_2$.
\end{thm}
%
%

In order for \texttt{HKP} to be effective, we need a procedure to
compute $p$.  When $\mathbb{S}$ is an 
%
integral $l$-monoid, we can check $(v,v')\in p(R)$ via a
variation of the congruence check, using a rewriting system as in
\cref{sec:cong-closure-l-monoids}.

\begin{prop}\label{prop:Preconalgocorrect}  
  Let $\mathbb L$ be an integral $l$-monoid and let
  $R\subseteq {\mathbb L}^X\times {\mathbb L}^X$ be a relation. 
  
  The set of rules $\mathcal R$ is defined as  $\{ v'\mapsto v\sqcup
  v' \mid (v,v')\in R \}$.\footnote{Whenever $v\le v'$, the rule can be omitted, since it is
    never applicable.}	 Rewriting steps are defined as in \cref{def:RewritingNF}. If the rewriting algorithm terminates, then for all $v,v'\in {\mathbb L}^X$, $(v,v')\in p(R)$ iff
  $\Downarrow\!v'\geq v$ (where, as usual, $\Downarrow\!v'$ denotes the normal form of $v'$).
\end{prop}

Observe that Theorems \ref{thm:rewritingterminates01maxtimes} and \ref{thm-termcond} guarantee termination for certain specific $l$-monoids. In particular, termination for the tropical semiring will be pivotal henceforward.

\subsection{Threshold Problem for Automata over the Tropical Semiring}
\label{sec:threshold}

Language inclusion for weighted automata over the tropical semiring is
not decidable, because language equivalence is not.  However, the
algorithm that we have introduced in the previous section can be used
to solve the so called \emph{threshold problem} over the tropical
semiring of natural numbers
$(\mathbb N_0\cup\{\infty\},\min, +,\infty, 0)$. The problem is to check whether for a given threshold
$T\in\mathbb N_0$, a vector of states of a weighted automaton
$v\in\mathbb (\mathbb N_0\cup\{\infty\})^X$ satisfies the threshold
$T$, i.e. $\llbracket v\rrbracket(w)\leq T$ for all $w\in A^*$. 

Note that this problem is also known as the \emph{universality
  problem}: universality for non-deterministic automata can be easily
reduced to it, by taking weight $0$ for each transition and setting
$T=0$ for the threshold.
\begin{wrapfigure}{R}{0pt}
  \vspace{-2cm}
  \scalebox{0.8}{
  \tikzset{elliptic state/.style={draw,ellipse}}
    \begin{tikzpicture}[x=2cm,y=-2cm,double distance=2pt]
    	\node[elliptic state] (d) at (3,0) {$t$} ;
      \node (final2) at (3.5,0){};
      \begin{scope}[->]
        
        \path[shorten <=1pt] 
          (d) edge[loop left] node[arlab,above] {$a,0$} (d) 
          		edge node[arlab,above] {$T$} (final2) ;
      \end{scope}
    \end{tikzpicture}
}
\end{wrapfigure}

This problem -- which is known to be $\mathsf{PSPACE}$-complete
\cite{Alma11whatsdecidable} -- can be reduced to language inclusion by adding a
new state $t$ with output $o(t)=T$ and a $0$ self-loop for each letter
$a\in A$. Then we check whether the language of $v$ includes the
language of the unit vector $e_t$.

It is worth to note that in
$(\mathbb N_0\cup\{\infty\},\min, +,\infty, 0)$ the ordering
$\sqsubseteq$ is actually $\geq$, the reversed ordering on natural
numbers. Therefore to solve the threshold problem, we need to check
$e_t\precsimu v$. 

The reader can easily concoct an example where $\texttt{HKP}$ may not terminate. 
However, it has already been observed in
\cite{Alma11whatsdecidable} that it is a sound reasoning
technique to replace every vector entry larger than $T$ by $\infty$.
%
%
%
To formalise this result, we will first introduce an abstraction
mapping $\mathcal A$ and then state our modified algorithm:

\begin{defi}
  Let a threshold $T\in\mathbb N_0$ be given. We define the
  abstraction
  $\mathcal A:\mathbb N_0\cup\{\infty\}\rightarrow\mathbb
  N_0\cup\{\infty\}$
  according to $\mathcal A(s) = s$ if $s\leq T$ and
  $\mathcal A(s) = \infty$ otherwise. The definition extends
  elementwise to vectors in $(\mathbb N_0\cup\{\infty\})^X$.
\end{defi}

With this definition, we call $\HKPA$, the algorithm obtained from
$\texttt{HKP}$ by replacing step \texttt{(3.4)} with the following:
\begin{codeNT}
 (3.4)  for all $a\in A$ 
          insert $(t_a(v_1'), \mathcal A(t_a(v_2')))$ into $\mathit{todo}$
\end{codeNT}

Now to check whether a certain vector $v$ satisfies the threshold of $T$, it is enough to run $\HKPA(e_t,v)$ where $e_t$ is the unit vector for $t$ as defined above.

The soundness of the proposed algorithm can be shown in essentially
the same way as for $\texttt{HKP}$ but using a novel up-to technique
to take care of the abstraction $\mathcal A$\full{ (see \cref{sec:appendix-applications})}. 
For the completeness, we need the following additional result.

\begin{lem}~ 
  \label{lems:Atacompat}
  For all vectors $v\in(\mathbb N_0\cup\{\infty\})^X$ it holds that
  (i) $\mathcal A(t_a(\mathcal A(v)))=\mathcal A(t_a(v))$; (ii)
  $\mathcal A(o(\mathcal A(v)))=\mathcal A(o(v))$.
\end{lem}

\begin{thm}
  \label{thm:modifalgcorrect}
  $\HKPA(e_t, v_1)$ always terminates. Moreover $\HKPA(e_t, v_1)$ returns true iff $\llbracket v_1\rrbracket(w)\leq T$ for all $w\in A^*$.
\end{thm}

\subsection{Exploiting Similarity}
\label{sec:exploiting-similarity}

For checking language inclusion of non-deterministic automaton it is
often convenient to precompute a \emph{similarity} relation that
allows to immediately skip some pairs of states
\cite{achmv:simulation-meets-antichains}. This idea can be readapted
to weighted automata over an $l$-monoid by using the following notion.


\begin{defi}
  \label{def:similarity-upto-le}
  Let $(X,o,t)$ be a weighted automaton. A relation
  $R\subseteq \mathbb{S}^X\times \mathbb{S}^X$ on unit vectors is
  called a \emph{simulation relation} whenever for all $(v,v')\in R$
  (i) $o(v)\sqsubseteq o(v')$; (ii) for all $a\in A$, there exists a
  pair $(u,u')$ that is a linear combination of vector pairs in $R$
  and furthermore $t_a(v)\sqsubseteq u$, $u'\sqsubseteq t_a(v')$.

  \emph{Similarity}, written $\preceq$, is the greatest simulation
  relation.
\end{defi}

\begin{lem}
  \label{simulationimplieslanguageinclusion}
  Simulation implies language inclusion, i.e. $\preceq$ is included in
  $\precsimu$.
\end{lem}

\begin{figure}[t]
\centering
\underline{\texttt{SIM} $(X,o,t)$}

\begin{codeNT}
(1)	$R:=\{(v,v')\in\mathbb S^X\times S^X\mid v,v'\mathit{\ are\ unit\ vectors}\}$
(2)	$R':=\emptyset$
(3)	for all $(v,v')\in R$
	(3.1)	if $o(v)\not\sqsubseteq o(v')$ then $R:=R\setminus\{(v,v')\}$
(4)	while $R\neq R'$
	(4.1)	$R':=R$
	(4.2)	for all $a\in A$
		(4.2.1)	for all $(v,v')\in R$
			(4.2.1.1)  $u := \bigsqcup\{v_1\cdot(v_2\rightarrow t_a(v'))\mid (v_1, v_2)\in R\}$
			(4.2.1.2)  if $t_a(v)\not\sqsubseteq u$ then $R:=R\setminus\{(v,v')\}$
(5)	return $R$
\end{codeNT}
\caption{Algorithm to  compute similarity  ($\preceq$) for  a weighted
  automaton $(X,o,t)$.}
\label{fig:sim}
\end{figure}

%

%
%

Similarity over an
$l$-monoid can be computed with the algorithm in \cref{fig:sim}.
Even though the relation is not symmetric, the method is conceptually
close to the traditional partition refinement algorithm to compute
bisimilarity. Starting from the cross-product of all states, the
algorithm first eliminates all pairs of states where the first state
does not have a smaller-or-equal output than the second one and then
continuously removes all pairs of states that do not meet the second
requirement for a simulation relation, until the relation does not
change anymore.

\begin{lem}
  \label{lem:greatestsim}
  $\texttt{SIM} $ computes $\preceq$.

\end{lem}

\begin{lem}
  \label{lem:sim-complexity}
  The runtime complexity of $\texttt{SIM} $ when applied to an
  automaton over state set $X$ and alphabet $A$ is
  polynomial, assuming constant time complexity for
  all semiring operations (supremum, multiplication, residuation).
\end{lem}

Once $\preceq$ is known, it can be exploited by \texttt{HKP} and
$\HKPA$.  To be completely formal in the proofs, it is convenient to
define two novel algorithm -- called \texttt{HKP'} and $\HKPA '$ --
which are obtained from \texttt{HKP} and $\HKPA$ by replacing step
\texttt{(3.2)} by
\begin{codeNT}
 (3.2)  if $(v_1',v_2')\in p'(R)$ then continue 
\end{codeNT}
where $p'(R)$ is defined for all
relations $R$ as $p'(R)=p(R\,\cup \preceq)$. The following two results state the correctness of the two algorithms.



\begin{lem}\label{lemmaHKP'}
  Let $\mathbb{S}$ be a semiring equipped with a precongruence
  $\sqsubseteq$.  Whenever \linebreak \texttt{HKP'}($v_1,v_2$)
  terminates, it returns true iff $\ v_1 \precsimu v_2$.
  \label{lem:langincalgcorrect}
\end{lem}

\begin{lem}
  \label{thm:HKPA'sound}
  $\HKPA '(e_t, v_1)$ always terminates. Moreover $\HKPA '(e_t, v_1)$ returns true iff $\llbracket v_1\rrbracket(w)\leq T$ for all $w\in A^*$.
\end{lem}

\subsection{An Exponential Pruning}

To illustrate the benefits of up-to techniques, we show an example
where $\HKPA '$ exponentially prunes the exploration space by
exploiting the technique $p'$. We compare $\HKPA'$ against $\HKA$ in
\cref{fig:naive}, that can be thought as an adaptation of the
algorithm proposed in \cite{Alma11whatsdecidable} to the
notation used in this paper.

\begin{figure}[t]
\centering
\underline{$\HKA (v_0)$}
	\begin{codeNT}
	(1)	$\mathit{todo} := \{v_0\}$
	(2)	$P:=\emptyset$
	(3)	while $\mathit{todo}\neq\emptyset$
		(3.1)	extract $v$ from $\mathit{todo}$
		(3.2)	if $v\in P$ then continue
		(3.3)	if $o(v)\not\leq T$ then return false
		(3.4)	for all $a\in A$ insert $\mathcal{A}(t_a(v))$ into $\mathit{todo}$
		(3.5)	insert $v$ into $P$
	(4)	return true
	\end{codeNT}
        \caption{Algorithm to check whether a vector $v_0$ of a
          weighted automata $(X,o,t)$ satisfies the threshold
          $T\in\mathbb N_0$}
\label{fig:naive}
\end{figure}

Consider the family of automata over the tropical semiring in
\cref{fig:exponentialspeedup} and assume that $T = n$. By taking as
initial vector $e_x \sqcup e_y$ (i.e., the vector mapping $x$ and $y$
to $0$ and all other states to $\infty$), the automaton clearly does
not respect the threshold, but this can be observed only for words
longer than $n$.

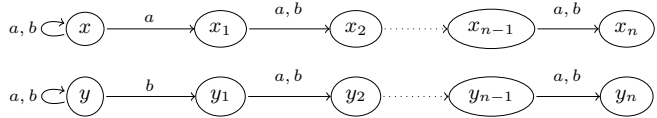
\begin{figure}
  \begin{center}
    \tikzset{every state/.style={minimum size=3em}}
    \tikzset{elliptic state/.style={draw,ellipse}}
    \scalebox{0.9}{\begin{tikzpicture}[x=2cm,y=-1cm,double distance=4pt]
      \node[elliptic state] (x) at (0,1) {$x$} ;
      \node[elliptic state] (x1) at (1,1) {$x_1$} ;
      \node[elliptic state] (x2) at (2,1) {$x_2$} ;
      \node[elliptic state] (x3) at (3,1) {$x_{n-1}$} ;
      \node[elliptic state] (x4) at (4,1) {$x_n$} ;
      \node[elliptic state] (y) at (0,2) {$y$} ;
      \node[elliptic state] (y1) at (1,2) {$y_1$} ;
      \node[elliptic state] (y2) at (2,2) {$y_2$} ;
      \node[elliptic state] (y3) at (3,2) {$y_{n-1}$} ;
      \node[elliptic state] (y4) at (4,2) {$y_n$} ;
      \begin{scope}[->]
        \path[shorten <=1pt] 
        (x) edge[loop left] node[arlab] {$a,b$} (x) 
        edge node[arlab,above] {$a$} (x1) ;
        \path[shorten <=1pt] 
        (x1) edge node[arlab,above] {$a,b$} (x2) ;
        \path[shorten <=1pt,dotted] 
        (x2) edge (x3) ;
        \path[shorten <=1pt] 
        (x3) edge node[arlab,above] {$a,b$} (x4) ;
        \path[shorten <=1pt] 
        (y) edge[loop left] node[arlab] {$a,b$} (y) 
        edge node[arlab,above] {$b$} (y1) ;
        \path[shorten <=1pt] 
        (y1) edge node[arlab,above] {$a,b$} (y2) ;
        \path[shorten <=1pt, dotted] 
        (y2) edge (y3) ;
        \path[shorten <=1pt] 
        (y3) edge node[arlab,above] {$a,b$} (y4) ;
      \end{scope}
    \end{tikzpicture}}
  \end{center}	

  \caption{Examples where $\HKPA'$ exponentially improves over $\HKA$. Output weight is always $0$, transition weight is
    always $1$.}
  \label{fig:exponentialspeedup}
\end{figure}

First, for $\HKA$ the runtime is exponential. This
happens, since every word up to length $n$ produces a different weight
vector. For a word $w$ of length $m$ state $x_i$ has weight $m$ iff
the $i$-last letter of the word is $a$, similarly state $y_i$ has
weight $m$ iff the $i$-last letter is $b$. All other weights are
$\infty$. For instance, the weights for word $aab$ are given below.

\begin{center}
  \begin{tabular}{|c|c|c|c|c|c|c|c|c|c|c|c|}
    $x$ & $x_1$ & $x_2$ & $x_3$ & $x_4$ & $\dots$ & $y$ & $y_1$ & 
    $y_2$ & $y_3$ & $y_4$ & $\dots$ \\ \hline
    $3$ & $\infty$ & $3$ & $3$ & $\infty$ & $\dots$ & 
    $3$ & $3$ & $\infty$ & $\infty$ & $\infty$ & $\dots$ 
  \end{tabular}
\end{center}

Now we compare with $\HKPA'$.  
Observe that $x_i\preceq x$, $y_i\preceq y$ for all $i$. (Remember
that since the order is reversed, a lower weight simulates a higher
weight.) Hence, we obtain rewriting rules that allow to replace an
$\infty$-entry in $x_i$ and $y_i$ by $m$ for all $i$. (Since both
entries $x$ and $y$ are $m$, we can always apply this rule.) In the
example above this leads to a vector where every entry is $3$.



Hence it turns out that for all words of the same length, the
corresponding vectors are all in the precongruence relation with each
other -- as they share the same normal form -- and we
only have to consider exactly one word of each length. Therefore, only
linearly many words are considered and the runtime is polynomial.

\section{Runtime Results for the Threshold Problem}
\label{sec:implementation}

We now discuss runtime results for the threshold problem for weighted
automata over the tropical semiring of the natural numbers. We compare
the following three algorithms: the algorithm without up-to technique
($\HKA$) algorithm in \cref{fig:naive}, the algorithm that works up-to
precongruence ($\HKPA$), explained in \cref{sec:threshold}, and the
algorithm that additionally exploits pre-computed similarity
($\HKPA'$), introduced in \cref{sec:exploiting-similarity}. This
precomputation step is relatively costly and is included in the
runtime results below.



We performed the following experiment: for certain values of $|X|$
(size of state set) and of $T$ (threshold) we generated random
automata. The alphabet size was randomly chosen between $1$ and $5$.
For each pair of states and alphabet symbol, the probability of having
an edge with finite weight is $90\%$. (We chose this high
number, since otherwise the threshold is almost never respected and
the algorithms return false almost immediately due to absence of a
transition for a given letter. With our choice instead, the algorithms
need many steps and the threshold is satisfied in $14\%$ of the cases.)
In case the weight is different from $\infty$, a random weight from
the set $\{0,\dots,10\}$ is assigned.

For each pair $(|X|, T)$ we generated 1000 automata. The
runtime results can be seen in \cref{fig:experiments}. We considered the
$50\%$, $90\%$ and $99\%$ percentiles: the $50\%$ percentile is the median
and the $90\%$ percentile means that $90\%$ of the runs were faster and
$10\%$ slower than the time given in the respective field. Analogously
for the $99\%$ percentile.

Apart from the runtime we also measured the size of the relation $R$
(or $P$ in the case of $\HKA$) and the size of the similarity
$\preceq$ (in case of $\HKPA'$). The program was written in C\# and
executed on an Intel Core 2 Quad CPU Q9550 at 2.83 GHz with 4 GB RAM,
running Windows 10. 

\begin{table}[!t]
  \centering
{\footnotesize
\begin{tabular}{|r|r||r|r|r||r|r|r||r|r|r|r|}
\hline
\multicolumn{2}{|c|}{} & \multicolumn{3}{c|}{Runtime (millisec.)} &
                                                                \multicolumn{3}{c|}{Size of $R$/$P$} & \multicolumn{3}{c|}{Size of $\preceq$} \\
\hline
$(|X|,T)$&algo&$50\%$&$90\%$&$99\%$&$50\%$&$90\%$&$99\%$&$50\%$&$90\%$&$99\%$\\
\hline
(3,10)&$\HKPA'$&2&8&20&5&14&33&0&2&4\\&$\HKPA$&1&3&14&5&14&34&-&-&-\\&$\HKA$&1&3&13&6&28&92&-&-&-\\\hline
(3,15)&$\HKPA'$&3&17&127&11&34&100&0&2&4\\&$\HKPA$&2&16&126&11&34&100&-&-&-\\&$\HKA$&2&17&90&18&119&373&-&-&-\\\hline
(3,20)&$\HKPA'$&6&65&393&18&70&174&0&2&4\\&$\HKPA$&4&64&466&18&71&192&-&-&-\\&$\HKA$&5&79&315&55&364&825&-&-&-\\\hline
(6,10)&$\HKPA'$&21&227&1862&18&106&302&0&2&12\\&$\HKPA$&8&217&1858&19&106&302&-&-&-\\&$\HKA$&9&286&2045&40&693&2183&-&-&-\\\hline
(6,15)&$\HKPA'$&90&2547&12344&65&353&750&0&2&11\\&$\HKPA$&84&2560&12328&65&353&750&-&-&-\\&$\HKA$&88&4063&20987&346&3082&7270&-&-&-\\\hline
(6,20)&$\HKPA'$&239&7541&59922&111&589&1681&0&3&11\\&$\HKPA$&234&7613&60360&111&589&1681&-&-&-\\&$\HKA$&253&16240&103804&702&6140&14126&-&-&-\\\hline
(9,10)&$\HKPA'$&274&9634&73369&98&582&1501&0&3&21\\&$\HKPA$&236&9581&72833&99&582&1501&-&-&-\\&$\HKA$&232&17825&99332&536&6336&14956&-&-&-\\\hline
(9,15)&$\HKPA'$&1709&71509&301033&256&1517&3319&0&3&19\\&$\HKPA$&1681&70587&301018&256&1517&3319&-&-&-\\&$\HKA$&919&112323&515386&1436&14889&28818&-&-&-\\\hline
(9,20)&$\HKPA'$&3885&168826&874259&407&2347&5086&0&3&20\\&$\HKPA$&3838&168947&872647&407&2347&5086&-&-&-\\&$\HKA$&1744&301253&1617813&2171&22713&48735&-&-&-\\\hline
(12,10)&$\HKPA'$&1866&93271&560824&247&1586&3668&0&7&31\\&$\HKPA$&1800&92490&560837&251&1586&3668&-&-&-\\&$\HKA$&1067&189058&889949&1342&18129&37387&-&-&-\\\hline
(12,15)&$\HKPA'$&5127&363530&1971541&423&3001&6743&0&7&35\\&$\HKPA$&5010&362908&1968865&423&3001&6743&-&-&-\\&$\HKA$&1418&509455&2349335&1672&27225&55627&-&-&-\\\hline
(12,20)&$\HKPA'$&15101&789324&3622374&744&4489&9027&0&6&32\\&$\HKPA$&15013&787119&3623393&744&4489&9027&-&-&-\\&$\HKA$&4169&1385929&4773543&3297&43756&80712&-&-&-\\\hline
\end{tabular}
\medskip
}\caption{Runtime results on randomly generated
  automata}\label{fig:experiments}
\vspace{-0.8cm}
\end{table}

First note that, as expected, $\HKPA$ and $\HKPA'$ always produce much
smaller relations than $\HKA$.  However, they introduce some overhead,
due to rewriting for checking $p(R)$, and due to the
computation of similarity, which is clearly seen for the $50\%$
percentile. 
However,
if we look at the larger parameters and at the $90\%$ and $99\%$
percentiles (which measure the worst-case performance), $\HKPA$ and
$\HKPA'$ gain the upper hand in terms of runtime.

Note also that while in the example above similarity played a large
role, this is not the case for the random examples. Here similarity
(not counting the reflexive pairs) is usually quite small. This means
that similarity does not lead to savings, only in very few cases does
the size of $R$ decrease for $\HKPA'$. But this also means that
the computation of $\preceq$ is not very costly and hence the runtime
of $\HKPA$ is quite similar to the runtime of $\HKPA'$.  We believe that for
weighted automata arising from concrete problems, the similarity
relation will usually be larger and promise better runtimes. Note also
that similarity is independent of the initial vector and the threshold
and if one wants to ask several threshold questions for the same
automaton, it has to be computed only once.

\section{Conclusion and Future Work}

In this work, we have investigated up-to techniques for weighted
automata, including methods to determine the congruence closure for
semimodules. 


\emph{Related work:} Related work on up-to techniques has already
been discussed in the introduction.
For the language equivalence problem for weighted automata we are
mainly aware of the algorithm presented in \cite{bls:conjugacy}, which
is a partition refinement algorithm and which already uses a kind of
up-to technique: it can eliminate certain vectors which arise as
linear combinations of other vectors. The paper
\cite{hu:forward-backward-quantitative} considers simulation for
weighted automata, but not in connection to up-to
techniques. 

Congruence closure for term rewriting has been investigated in
\cite{Cyrluk96onshostaks}.

Our examples mainly involved the tropical semiring (and related
semirings). Hence there are relations to work by Aceto et al.
\cite{Aceto2003417} who presented an equational theory for the
tropical semiring and related semirings, as well as Gaubert et al.
\cite{gaubert1997} who discuss several reasons to be interested in the
tropical semiring and present solution methods for several types of
linear equation systems.


\emph{Future work:}
As we have seen in the experiments on the threshold problem, our
techniques greatly reduce the size of the relations. However, the
reduction in runtime is less significant, which is due to the overhead
for the computation of similarity and the rewriting procedure. There
is still a substantial improvement for the worst-case running times
(90\% and 99\% percentiles).  So far, the algorithms, especially
algorithm \texttt{SIM} for computing similarity, are not very
sophisticated and we believe that there is further potential for
optimization.

\smallskip

\textbf{Acknowledgements:} We would like to thank Pawe{\l}
Soboci\'{n}ski and Damien Pous for interesting discussions on the
topic of this paper. Furthermore we express our thanks to Issai Zaks
for his help with the runtime results.

\short{
\bibliographystyle{plain} 
\bibliography{references}
}

\full{
}

\full{

\appendix



\section{Proofs}
\label{sec:proofs}


Here we give proofs for all lemmas and propositions where we have
omitted the proofs in the article.

\subsection{Preliminaries}

We will prove some properties $l$-monoids and residuation which will
play an important role in proving our main results.

\begin{lem}\label{lem:alternatedefofvectorresiduation}
  Let $(\mathbb L, \sqcup, \cdot, 0,1)$ be an $l$-monoid and $v,v'$ be
  $n$-dimensional $\mathbb L$-vectors. Then it holds that
  $$v\rightarrow v'=\bigsqcap\{v[i]\rightarrow v'[i]\mid 1\leq i\leq
  n\}$$
\end{lem}

\begin{proof}
  We define
  $(v\Rightarrow v') := \bigsqcap\{v[i]\rightarrow v'[i]\mid 1\leq
  i\leq n\}$.

  First we will show that $v\Rightarrow v'\sqsubseteq v\rightarrow v'$.  In
  order to prove this, we will show that
  $v\cdot(v\Rightarrow v')\leq v'$, i.e.
  $(v\Rightarrow v')\in\{\ell\in\mathbb L\mid v\cdot \ell\sqsubseteq v'\}$:
  \begin{equation*}
    \begin{aligned}
      &v\cdot(v\Rightarrow v')=\begin{pmatrix}v[1]\\v[2]\\\vdots\\v[n]\end{pmatrix}\cdot(v\Rightarrow v')=\begin{pmatrix}v[1]\cdot(v\rightarrow v')\\v[2]\cdot(v\rightarrow v')\\\vdots\\v[n]\cdot(v\rightarrow v')\end{pmatrix}\sqsubseteq\begin{pmatrix}v[1]\cdot(v[1]\rightarrow v'[1])\\v[2]\cdot(v[2]\rightarrow v'[2])\\\vdots\\v[n]\cdot(v[n]\rightarrow v'[n])\end{pmatrix}\\
      \sqsubseteq&\begin{pmatrix}v'[1]\\v'[2]\\\vdots\\v'[n]\end{pmatrix}=v'
    \end{aligned}
  \end{equation*}
  Next we need to show that $v\Rightarrow v'\sqsupseteq v\rightarrow v'$.  It
  suffices to show that $v\Rightarrow v'$ is an upper bound of the set
  $\{\ell\in\mathbb L\mid v\cdot \ell\sqsubseteq v'\}$. Since
  $v\Rightarrow v'$ is the greatest lower bound of
  $\{v[i]\rightarrow v'[i]\mid 1\leq i\leq n\}$, it is enough to show
  that every element of the first set and every element of the second
  set are in relation. Hence take $\ell\in \mathbb L$ with
  $v\cdot \ell \sqsubseteq v'$ and an index $i$. Since $v\cdot \ell \sqsubseteq v'$
  (componentwise), it holds that $v[i]\cdot \ell \sqsubseteq v'[i]$ for every
  $i$. Hence $\ell \sqsubseteq v[i]\sqsubseteq v'[i]$.

\end{proof}


\begin{lem}
  In an $l$-monoid $\mathbb L$, \label{lem:monotonicity}
  \begin{lemlist}
  \item Whenever $\ell_1\sqsubseteq \ell_2$ it follows that
    $\ell\cdot \ell_1\sqsubseteq \ell\cdot \ell_2$ and
    $ \ell_1\cdot\ell\sqsubseteq \ell_2\cdot \ell$ for all $l$-monoid
    elements $\ell,\ell_1,\ell_2$.
  \item For all $l$-monoid vectors $v,v'\in \mathbb{L}^Y$ and matrices
    $M\in \mathbb L^{X\times Y}$ it holds that $v\sqsubseteq v'$ implies
    $Mv\sqsubseteq Mv'$.
  \item if $\mathbb L$ is integral, $\ell\cdot \ell'\sqsubseteq \ell'$ and
    $\ell\cdot \ell'\sqsubseteq \ell$.
  \item For all $\ell_1,\ell_2,\ell_3\in\mathbb L$, we have
    $(\ell_1\rightarrow \ell_2)\cdot \ell_3\sqsubseteq
    \ell_1\rightarrow(\ell_2\cdot \ell_3)$.\label{lem:monotonicity-4}
  \item If $\ell_2\sqsubseteq \ell_3$ it follows that
    $(\ell_1\rightarrow \ell_2)\sqsubseteq(\ell_1\rightarrow \ell_3)$.
  \item It holds that
    $\ell_3 \sqsubseteq \ell_1\to \ell_2 \iff \ell_1\cdot \ell_3 \sqsubseteq
    \ell_2$.\label{lem:monotonicity-6}
  \item It holds that
    $\ell_1\to (\ell_2\sqcup \ell_3) \sqsupseteq (\ell_1\to \ell_2) \sqcup
    (\ell_1\to \ell_3)$. \label{lem:monotonicity-7}
	\item It holds that $\ell_1\cdot (\ell_1\to \ell_2) \sqsubseteq \ell_2$ \label{lem:monotonicity-8}
  \end{lemlist}
\end{lem}

\begin{proof}
\begin{lemlist}
 \item $$\ell\cdot \ell_2=\ell\cdot(\ell_1\sqcup \ell_2)=\ell\cdot
\ell_1\sqcup \ell\cdot \ell_2\sqsupseteq \ell\cdot \ell_1$$
and
$$\ell_2\cdot \ell=(\ell_1\sqcup \ell_2)\cdot\ell=
\ell_1\cdot\ell\sqcup  \ell_2\cdot\ell\sqsupseteq  \ell_1\cdot\ell.$$
\item The second part follows directly, because 
  $$(Mv)[i]=\bigsqcup_{j\in Y} M[i,j]\cdot v[j] \underset{v[j]\sqsubseteq
    v'[j]}\sqsubseteq \bigsqcup_{j\in Y} M[i,j]\cdot v'[j] = (Mv')[i]$$
\item This follows directly from monotonicity,
  $\ell\cdot \ell'\sqsubseteq\top\cdot \ell'=1\cdot \ell'=\ell'$ and
  $\ell\cdot \ell'\sqsubseteq \ell\cdot\top=\ell\cdot1=\ell$.
\item We first compute:
  $$(\ell_1\rightarrow \ell_2)\cdot
  \ell_3=\bigsqcup\{\ell\in\mathbb L\mid \ell_1\cdot \ell\sqsubseteq
  \ell_2\}\cdot \ell_3=\bigsqcup\{\ell\cdot \ell_3\mid \ell_1 \cdot
  \ell\sqsubseteq \ell_2\}$$  and:
  $$\ell_1\rightarrow(\ell_2\cdot \ell_3)=\bigsqcup\{\ell\in\mathbb L\mid \ell_1\cdot \ell\sqsubseteq \ell_2\cdot \ell_3\}$$
  Now we obtain:
  $$\ell_1 \cdot \ell\sqsubseteq \ell_2\Rightarrow \ell_1\cdot \ell\cdot \ell_3\sqsubseteq \ell_2\cdot \ell_3$$
  and therefore
  $\{\ell\cdot \ell_3\mid \ell_1 \cdot \ell\sqsubseteq \ell_2\}\subseteq
  \{\ell\in\mathbb L\mid \ell_1\cdot \ell\sqsubseteq \ell_2\cdot \ell_3\}$.        
\item Obviously,
  $\{\ell\in\mathbb L\mid \ell_1\cdot \ell\sqsubseteq
  \ell_2\}\subseteq\{\ell\in\mathbb L\mid \ell_1\cdot \ell\sqsubseteq
  \ell_3\}$ and therefore:
  $$\ell_1\rightarrow \ell_2=\bigsqcup\{\ell\in\mathbb L\mid
  \ell_1\cdot \ell\sqsubseteq \ell_2\}\sqsubseteq\bigsqcup\{\ell\in\mathbb L\mid
  \ell_1\cdot \ell\sqsubseteq \ell_3\}=\ell_1\rightarrow \ell_3$$.
\item Whenever $\ell_3 \sqsubseteq \ell_1\to \ell_2$, then
  $\ell_1\cdot \ell_3 \sqsubseteq \ell_1\cdot (\ell_1\to \ell_2) = \ell_1\cdot
  \bigsqcup \{\ell\in\mathbb L \mid \ell_1\cdot \ell \sqsubseteq \ell_2 \} =
  \bigsqcup \{\ell_1\cdot \ell\in\mathbb L \mid \ell_1\cdot \ell \sqsubseteq
  \ell_2 \} \sqsubseteq \ell_2$.

  Whenever $\ell_1\cdot \ell_3 \sqsubseteq \ell_2$ we have that
  $\ell_1\to \ell_2 = \bigsqcup \{\ell\in\mathbb L\mid \ell_1\cdot
  \ell \sqsubseteq \ell_2\} \sqsupseteq \ell_3$,
  since $\ell_3$ is an element of the set.
\item Because of \cref{lem:monotonicity-6} it suffices to show that
  $\ell_1\cdot ((\ell_1\to \ell_2)\sqcup (\ell_1\to \ell_3)) =
  \ell_1\cdot (\ell_1\to \ell_2) \sqcup \ell_1\cdot (\ell_1\to \ell_3)
  \sqsubseteq \ell_2\sqcup \ell_3$.
\item $\ell_1\cdot (\ell_1\to \ell_2) =\ell_1\cdot\bigsqcup\{\ell'\mid \ell_1\cdot\ell'\sqsubseteq\ell_2\}=\bigsqcup\{\ell_1\cdot\ell'\mid \ell_1\cdot\ell'\sqsubseteq\ell_2\}\sqsubseteq \bigsqcup\{\ell_2\}=\ell_2$
\end{lemlist}
\end{proof}

\subsection{Congruence Closure}
\label{sec:appendix-congruence-closure}

\subsection*{Congruence Closure for Rings}

\begin{lem}~
  \label{lem:uRmodulerUcong}
  \begin{lemlist}
  \item Let $\mathbb{I}$ be a ring. Let
    $R\subseteq \mathbb{I}^X\times\mathbb{I}^X$ be a congruence. Then
    $u(R)$ is a module.\label{lem:uRmodule}
  \item Let $\mathbb{I}$ be a ring. Let $U\subseteq\mathbb{I}^X$ be a
    module. Then $r(U)$ is a congruence.\label{lem:rUcong}
  \end{lemlist}
\end{lem}

\begin{proof}~
\begin{lemlist}
\item We will show that $u(R)$ contains all vectors generated via
  linear combination from $u(R)$, making $u(R)$ a generating set for
  itself, i.e. a module.
  \begin{itemize}
  \item Let $v''\in u(R)$ then there must be $v,v'\in\mathbb{I}^X$ such
    that $(v,v')\in R$ and $v-v'=v''$. Since $R$ is congruence, it
    follows that $(v\cdot s, v'\cdot s)\in R$ for any $s\in\mathbb{I}$.
    This means that $v\cdot s - v'\cdot s\in u(R)$, distributivity now
    proves $(v-v')\cdot s\in u(R)$, i.e. $v''\cdot s\in u(R)$.
  \item Let $v''_1, v''_2\in u(R)$. Then there must be
    $v_1, v'_1, v_2, v'_2\in\mathbb{I}^X$ such that $v''_i=v_i-v'_i$
    and $(v_i,v'_i)\in R$ for $i = 1,2$.  Since $R$ is a congruence,
    it follows that $(v_1+v_2, v'_1+v'_2)\in R$. Thus,
    $(v_1+v_2) - (v'_1+v'_2)\in u(R)$. Commutativity of addition
    yields $(v_1-v'_1)+(v_2-v'_2)\in u(R)$, i.e.
    $v''_1+v''_2 \in u(R)$.
  \end{itemize}
\item 
  \begin{itemize}
  \item \emph{Reflexivity:} Let any $v\in U$ be given, then
    $v\cdot 0\in U$, so the $0$-vector is in $U$. For any given
    $v\in\mathbb{I}^X$, $v-v=0$, thus $(v,v)\in r(U)$.
  \item \emph{Symmetry:} Let $(v,v')\in r(U)$, then $v-v'\in U$. Since
    $U$ is a module, $(v-v')\cdot (-1)\in U$ and thus
    $-v+v'=v'-v\in U$, therefore $(v', v)\in r(U)$.
  \item \emph{Transitivity:} Let $(v,v')\in r(U)$, $(v',v'')\in r(U)$,
    then $v-v'\in U$ and $v'-v''\in U$. Since $U$ is a module,
    $(v-v')+(v'-v'')\in U$ and thus $v-v''\in U$. Therefore
    $(v,v'')\in R(U)$.
  \item \emph{Addition:} Let $(v_1,v_2)\in r(U)$ and
    $(v_1',v_2')\in r(U)$, then $v_1-v_2\in U$ and $v_1'-v_2'\in U$.
    Therefore $(v_1-v_2)+(v_1'-v_2')\in U$. Commutativity yields
    $(v_1+v_1')-(v_2+v_2')\in U$, i.e. $(v_1+v_1', v_2+v_2')\in r(U)$.
  \item \emph{Multiplication:} Let $(v,v')\in r(U)$ and
    $s\in\mathbb{I}$. Then $v-v'\in U$. Since $U$ is a module,
    $(v-v')\cdot s\in U$, distributivity yields
    $v\cdot s - v'\cdot s\in U$ and thus per definition
    $(v\cdot s,v'\cdot s)\in r(U)$.
  \end{itemize}      
\end{lemlist}
\end{proof}

\subsection*{Proof of \cref{prop:congclosurerings}}

\begin{proof}
  We first define the following two functions:
  \begin{itemize}
  \item
    $u: \mathcal P(\mathbb{I}^X\times\mathbb{I}^X)\rightarrow \mathcal
    P(\mathbb{I}^X)$ with $u(R)=\{v-v'\mid(v,v')\in R\}$.
  \item
    $r:\mathcal P(\mathbb{I}^X)\rightarrow\mathcal P(\mathbf
    R^X\times\mathbb{I}^X)$
    with $r(U)=\{(v,v')\mid v-v'\in U\}$, where
    $U\subseteq \mathbb{I}^X$, i.e., $U$ is a set of
    $\mathbb{I}$-vectors.
  \end{itemize}
  According to \cref{lem:uRmodulerUcong} in \cref{sec:proofs} we know
  that if $R$ is a congruence, then $u(R)$ is a submodule and if $U$
  is a submodule then $r(U)$ is a congruence.

  Observe that $R\subseteq R'$ implies $u(R)\subseteq u(R')$ via
  definition of $u$ and $r(U)\subseteq r(U')$ whenever
  $U\subseteq U'$, by definition of $r$. 

  Observe furthermore that $r(u(R))\subseteq c(R)$ holds, because if
  $(v_1,v_2)\in r(u(R))$, then there exists a $(v_1',v_2')\in R$ such
  that $v_1-v_2=v_1'-v_2'$, hence $v_1'-v_1=v_2'-v_2$.  Now
  $(v_1'-v_1,v_1'-v_1)\in c(R)$ due to reflexivity and thus we obtain:
  $(v_1,v_2)+(v_1'-v_1,v_1'-v_1)=(v_1,v_2)+(v_1'-v_1,
  v_2'-v_2)=(v_1+v_1'-v_1, v_2+v_2'-v_2)=(v_1',v_2')$,
  hence $(v'_1,v'_2)\in c(R)$.  Thus, $r(u(R))\subseteq c(R)$, proving
  also that congruences are fixed points of the monotone function
  $r\circ u$, since $r(U)$ is always a congruence and for every
  congruence $R$ it holds that $c(R)=R$.
	
  Now we can observe that the module generated by a set of vectors is
  the smallest module that contains this set and the congruence
  closure of a relation is the smallest congruence closed relation
  containing that relation. 

  We will now show $r([U_R]) = r([u(R)]) = c(R)$, thus proving the
  statement of the proposition. We have $u(R)\subseteq u(c(R))$ and we
  know that $u(c(R))$ is a submodule from \cref{lem:uRmodule},
  hence the submodule generated by $u(R)$ is included in $u(c(R))$,
  i.e. $[u(R)]\subseteq u(c(R))$.  Therefore,
  $r([u(R)])\subseteq r(u(c(R)))=c(R)$, and since
  \cref{lem:rUcong} shows that $r$ applied to a submodule yields
  a congruence and we have $R\subseteq r([u(R)])\subseteq c(R)$, the
  second inclusion is indeed an equality. 
\end{proof}

\subsection*{Congruence Closure for $l$-Monoids}

\subsection*{Proof of \cref{lem:rewriting}}

\begin{proof}~
  \begin{lemlist}
  \item Assume that $v\leadsto v'$, via rule $l\mapsto r$, and
    $v\sqsubseteq w$. Then
    $v' = v\sqcup r\cdot (l\to v) \sqsubseteq w\sqcup r\cdot (l\to w) =: w'$.
    Hence either $w\leadsto w'$ or $w=w'$ and in this case $w\sqsupseteq v'$.
  \item Assume that $v\leadsto v'$, via rule $l\mapsto r$. Define
    $u := (v\sqcup w) \sqcup r\cdot (l\to (v\sqcup w)) \sqsupseteq (v\sqcup w)
    \sqcup r\cdot ((l\to v)\sqcup (l\to w)) = (v \sqcup r\cdot (l\to
    v)) \sqcup (w\sqcup r\cdot (l\to w)) \sqsupseteq v'\sqcup w$. The first
    inequality is due to \cref{lem:monotonicity-7}.

    Now either $v\sqcup w \leadsto u$ or
    $v\sqcup w = u \sqsupseteq v'\sqcup w$. Since we also have that
    $v\sqcup w \sqsubseteq v'\sqcup w$, this implies $v\sqcup w = v'\sqcup w$.
  \end{lemlist}
\end{proof}

\subsection*{Proof of \cref{lem:ChurchRosser}}

\begin{proof}
  Assume that $v\leadsto v_1$ and $v\leadsto v_2$.  We set
  $v_0^a = v$, $v_1^a = v_1$ and consider a sequence of rewriting
  steps
  $v_1^a \leadsto v_2^a \leadsto \dots \leadsto v_n^a \not\leadsto$
  that leads to normal form $v_n^a$. 

  We now construct a seqence of vectors $v_1^b,\dots,v_{n+1}^b$ where
  $v_1^b = v_2$, $v_i^b \leadsto v_{i+1}^b$ or $v_i^b = v_{i+1}^b$,
  and $v_{i+1}\sqsupseteq v_i^a$.

  Given $v_i^b$ with $i\ge 1$, \cref{lem:rewriting-2} guarantees the
  existence of $v_{i+1}^b \sqsupseteq v_i^a$ with $v_i^b\leadsto v_{i+1}^b$,
  or $v_i^b \sqsupseteq v_i^a$. In the latter case we set $v_{i+1}^b = v_i^b$.

  Since $v\leadsto^* v_{n+1}^b$ and $v_n^a$ is a normal form, it must
  hold that $v_{n+1}^b \le v_n^a$. We also know from above that
  $v_{n+1}^b \sqsupseteq v_n^a$, hence $v_{n+1}^b = v_n^a$. Hence, this is the
  vector $v'$ which is reachable from both $v_1$ and $v_2$ and which
  proves the local Church-Rosser property.

\end{proof}

\subsection*{Proof of \cref{lem:prepcorrectness}}

\begin{proof}
  Assume that $v\leadsto v'$ via rule $l\mapsto r$. Hence
  $v' = v\sqcup r\cdot (l\to v)$. Let
  $v'' := v\cdot \ell\sqcup r\cdot (l\to v\cdot \ell)$. By adapting
  the proof of \cref{lem:monotonicity-4} to vectors we can show that
  $v'' \sqsupseteq v\cdot \ell\sqcup r\cdot (l\to v)\cdot \ell = (v\sqcup
  r\cdot (l\to v))\cdot \ell = v'\cdot \ell$.
  Hence either $v\cdot \ell \leadsto v''$ or $v\cdot \ell = v''$. In
  the latter case $v\cdot \ell \sqsupseteq v'\cdot \ell$ (since
  $v''\sqsupseteq v'\cdot\ell$), but also $v\cdot \ell \sqsubseteq v'\cdot \ell$
  (since $v\sqsubseteq v'$). Hence
  $v\cdot \ell = v'\cdot \ell$.

  Now assume that
  $v=v_0\leadsto v'=v_1\leadsto v_2\leadsto \dots\leadsto v_n =
  \,\Downarrow\! v$.
  We construct a sequence of vectors $w_i$ where $w_0 = v\cdot \ell$,
  $w_i \sqsupseteq v_i\cdot \ell$ and either $w_i\leadsto w_{i+1}$ or
  $w_i = w_{i+1}$. 

  Given $w_i \sqsupseteq v_i\cdot \ell$ and $v_i\leadsto v_{i+1}$. We know
  that one of the following two cases holds:
  \begin{itemize}
  \item there exists $v''$ such that
    $v_i\cdot \ell\leadsto v'' \sqsupseteq v_{i+1}\cdot \ell$: now, since
    $w_i \sqsupseteq v_i\cdot \ell$, we know due to \cref{lem:rewriting-2}
    that there exists $w_{i+1}$ such that
    $w_i\leadsto w_{i+1} \sqsupseteq v''\sqsupseteq v_{i+1}\cdot \ell$ or
    $w_i \sqsupseteq v''$. In the second case we set $w_{i+1} = w_i$.
  \item \emph{or} $v_i\cdot \ell = v_{i+1}\cdot \ell$: again we set
    $w_{i+1} = w_i$ and obtain
    $w_{i+1} = w_i \sqsupseteq v_i\cdot\ell = v_{i+1}\cdot \ell$.
  \end{itemize}
  Hence $w_n \sqsupseteq v_n =\,\Downarrow\! v$ and via monotonicity we obtain
  $\Downarrow\! (v\cdot \ell) = \Downarrow\! w_0 = \Downarrow\! w_n \sqsupseteq
  \Downarrow\! v_n = \Downarrow\! v$.

	
\end{proof}

\subsection*{Proof of \cref{thm:RewritingCorrectness}}

\begin{proof}
  It suffices to show that $\Downarrow\! v\sqsupseteq v'$ and
  $\Downarrow\! v'\sqsupseteq v$, because if $v\sqsubseteq \Downarrow\! v'$, then
  $\Downarrow\! v\sqsubseteq \Downarrow\! v'$ ($\Downarrow$ is monotone and
  idempotent) and vice-versa.  We prove this via rule induction (cf.\
  rules in \cref{tab:proof-rules}).
  \begin{description}
  \item[(\textsc{Rel})] If we find that $v \ c(R)\ v'$ because
    $v\ R\ v'$, then there are rules $v\mapsto v\sqcup v'$ and
    $v'\mapsto v\sqcup v'$. 

    Hence $v$ rewrites to
    $v\sqcup (v\sqcup v')\cdot (v\to v) \ge v\sqcup v'$, since
    $v\to v \sqsupseteq 1$ (or $v$ can not be rewritten via this rule). Hence
    either $v$ is rewritten to a vector larger or equal $v'$ or
    $v \sqsupseteq v'$ holds. Hence $\Downarrow\! v \sqsupseteq v'$.

    Analogously one can show $\Downarrow\! v' \sqsupseteq v$.
  \item[(\textsc{Refl})] If we find that $v \ c(R)\ v'$ because of
    reflexivity (i.e.  $v=v'$), then trivially
    $\Downarrow\! v=\Downarrow\! v'$
  \item[(\textsc{Sym})] If we find that $v \ c(R)\ v'$ because of
    symmetry, then we already know from the induction hypothesis that
    $\Downarrow\! v\sqsupseteq v'$ and $\Downarrow\! v'\sqsupseteq v$.
  \item[(\textsc{Trans})] If we find that $v_1 \ c(R)\ v_3$ because of
    transitivity, i.e. $v_1 \ c(R)\ v_2$ and $v_2 \ c(R)\ v_3$, we
    know from the induction hypothesis that $\Downarrow\! v_1\sqsupseteq v_2$
    and $\Downarrow\! v_2\sqsupseteq v_1$ as well as $\Downarrow\! v_3\sqsupseteq v_2$
    and $\Downarrow\! v_2\sqsupseteq v_3$. In particular, we have
    $v_1\leq \Downarrow\! v_2\sqsubseteq\Downarrow\! v_3$ and
    $v_3\leq \Downarrow\! v_2\sqsubseteq \Downarrow\! v_1.$
  \item[(\textsc{Sca})] If we find that
    $v\cdot \ell \ c(R)\ v'\cdot \ell$ because $v \ c(R)\ v'$, then
    $v\sqsubseteq\Downarrow\! v'$ and therefore, using
    \cref{lem:prepcorrectness}
    $$ v\cdot \ell\sqsubseteq (\Downarrow\! v)\cdot \ell\sqsubseteq (\Downarrow\! v')\cdot
    \ell\sqsubseteq\Downarrow\!( v'\cdot \ell)$$
  \item[(\textsc{Plus})] If we find that
    $\overline v\sqcup v \ c(R)\ \overline v'\sqcup v'$ because of
    $\overline v \ c(R)\ \overline v'$ and $v \ c(R)\ v'$, then
    $$v\sqcup \overline v \sqsubseteq (\Downarrow\! v')\sqcup (\Downarrow\! 
    \overline v')\sqsubseteq \Downarrow\!(v'\sqcup \overline v'),$$
    due to the monotonicity of $\Downarrow$.
  \end{description}
\end{proof}

\subsection*{Proof of \cref{prop:rewritingtogreater}}

\begin{proof}
  We prove this via rule induction (cf.\ rules in
  \cref{tab:proof-rules}). 

  \begin{description}
  \item[(\textsc{Rel})] If we find that $v \ c(R)\ \overline v$
    because $v\ R\ \overline v$, then there are rules
    $v\mapsto v\sqcup \overline v$ and
    $\overline v\mapsto v\sqcup \overline v$. As in the proof of
    \cref{thm:RewritingCorrectness} we obtain that $v$ rewrites to a
    vector larger or equal $v\sqcup \overline v$ in one step or that
    $v$ itself has this property. Analogously for $\overline v$.
  \item[(\textsc{Refl})] If we find that $v \ c(R)\ \overline v$
    because of reflexivity (i.e. $\overline v=v$), then no rewriting
    step is needed.
  \item[(\textsc{Sym})] If we find that $v \ c(R)\ \overline v$
    because of symmetry, then we already know this from the induction
    hypothesis because the property we want to prove is symmetric.
  \item[(\textsc{Trans})] If we find that $v_1 \ c(R)\ v_3$ because of
    transitivity, i.e. $v_1 \ c(R)\ v_2$ and $v_2 \ c(R)\ v_3$, we
    know inductively: $v_1\leadsto^* v_1'\sqsupseteq v_1\sqcup v_2$,
    $v_2\leadsto ^*v_2'\sqsupseteq v_2\sqcup v_3$.

    Now, due to \cref{lem:rewriting-3}
    $v_1\sqcup v_2\leadsto^* u \sqsupseteq v_1\sqcup v'_2 \sqsupseteq v_1\sqcup
    v_2\sqcup v_3 \sqsupseteq v_1\sqcup v_3$.
    Furthermore since $v'_1\sqsupseteq v_1\sqcup v_2$ we know from
    \cref{lem:rewriting-2} that
    $v_1\leadsto^* v'_1\leadsto^* v''_1 \sqsupseteq u$. Combined, we obtain
    $v_1\leadsto^* v''_1 \sqsupseteq v_1\sqcup v_3$.

    For $v_3$ the proof is analogous.


  \item[(\textsc{Sca})] If we have
    $v_1\cdot \ell \ c(R)\ v_2\cdot \ell$ because $v_1 \ c(R)\ v_2$
    then $v_1\leadsto v_1'\sqsupseteq v_1\sqcup v_2$ and
    $v_2\leadsto v_2'\sqsupseteq v_1\sqcup v_2$. Thus, using
    \cref{lem:prepcorrectness}
    $$ v_1\cdot \ell\leadsto^* v_1''\sqsupseteq v_1'\cdot \ell\sqsupseteq(v_1\sqcup v_2)\cdot \ell=v_1\cdot \ell\sqcup v_1'\cdot \ell$$
    $$ v_2\cdot \ell\leadsto^* v_2''\sqsupseteq v_2'\cdot \ell\sqsupseteq(v_1\sqcup v_2)\cdot \ell=v_1\cdot \ell\sqcup v_1'\cdot \ell$$
  \item[(\textsc{Plus})] If we have
    $v_1\sqcup v_2 \ c(R)\ v_1'\sqcup v_2'$ because of
    $v_1 \ c(R)\ v_1'$ and $v_2 \ c(R)\ v_2'$, then
    $v_1\leadsto^* v_1''\sqsupseteq v_1\sqcup v_1'$ and
    $v_2\leadsto^* v_2''\sqsupseteq v_2\sqcup v_2'$ and we obtain with
    \cref{lem:rewriting}:
    $$v_1\sqcup v_2\leadsto^* v\sqsupseteq v_1''\sqcup v_2''\sqsupseteq (v_1\sqcup
    v_1')\sqcup(v_2\sqcup v_2')=(v_1\sqcup v_2)\sqcup(v_1'\sqcup
    v_2').$$ Analogously for $v'_1,v'_2$.
  \end{description}
\end{proof}

\subsection*{Proof of \cref{cor:terminationgeneral}}

\begin{proof}
  Take
  $\overline v = \bigsqcup \{\hat{v} \mid v \leadsto^* \hat{v} \}$. By
  \cref{prop:rewritingtogreater} if $v \ c(R)\ \overline v$, then
  $v \leadsto^* \overline v$. Since $\leadsto$ is irreflexive, 
  $\overline v \not \leadsto$ and $\overline v$ is in normal form.

  If we assume that each rule that is applicable is applied at one
  point (or rendered unapplicable by other rule applications), it is
  sufficient to know that there exists one rewriting sequence reaching
  $\overline{v}$ from $v$.  If we decide to apply a rule, that was
  applied in this specific sequence, at a later point of time, either
  we have already exceeded the corresponding vector in the original
  rewriting sequence via other rule applications, or the rule can
  still be applied with the same or a greater multiplicand, leading
  again to a a larger vector.

  Hence every sequence of rewriting steps will eventually reach
  $\overline v$.
\end{proof}

\subsection*{Termination for Specific $l$-Monoids}

\subsection*{Proof of \cref{thm:rewritingterminates01maxtimes}}

\begin{proof}
  We show this via contradiction. So we assume the algorithm does not
  terminate, i.e. there exists an infinite sequence of rewriting steps
  starting from a vector $v$. If that is the case, observe that there
  are some indices such that the corresponding entries in the vector
  increase infinitely often. We can assume that from the beginning,
  there are only indices that increase infinitely often or do not
  increase at all, because otherwise, we can apply rules until this is
  true and use the resulting vector as the new starting vector.
  Equivalently we can assume that each rule is applied infinitely
  often, by applying rules until no rule that can only be applied
  finitely often can ever get applied anymore and then removing all
  these rules from the rule system.
	
  We call the initial vector $v$ and the rule system $\mathcal R$. The
  sequence of intermediate rewriting results is a sequence
  $v_0, v_1, \ldots$ where $v_i\leadsto v_{i+1}$ for all
  $i\in\mathbb N_0$ in a single rewriting step. Taking a look at the
  history of a specific component $v[j]$ of $v$, we can observe that
  in each rewriting step applying a rule $l\mapsto r$, $v[j]$ either
  does not change or is rewritten to $\frac{r[j]\cdot v[j']}{l[j']}$
  for a vector-index $j'$. In fact, we choose the index $j'$ which
  minimizes that quotient. Inductively, we obtain that at any given
  rewriting step $i$,
  $$v_i[j]=\frac{r_n[j]\cdot r_{n-1}[j_{n-1}] \cdot\ldots\cdot 
    r_1[j_1]\cdot v[j_0]}{l_n[j_{n-1}]\cdot
    l_{n-1}[j_{n-2}]\cdot\ldots\cdot l_1[j_0]}$$
  where $j_0,\ldots, j_{n}$ are vector-indices and
  $l_1\mapsto r_1, \ldots, l_n\mapsto r_n$ are rules. The maximum
  index of rules is not the same as $i$, because only those rules are
  multiplied that really contributed to the value of $v_i[j]$
  directly, thus, $n$ might be smaller than $i$.  Note that we used
  the fact that multiplication with $1$ in order to apply a rule
  cannot happen, since then the rule could never be applied again. (In
  this case we say that the rule has been applied maximally.) If a
  rule was used maximally instead, it would not necessarily contribute
  a factor of the type $\frac {r[\overline i]}{l[\overline j]}$. Note
  that this representation is unique and we say $v_i[j]$ is based on
  $v[j_0]$ if it can be written as above.
  
  Let $N$ be the dimension of $v$. At any given time $i$, there are at
  most $N$ different entries in $v_i$ and each entry in $v_{i+1}$ is
  obtained by multiplying one factor of the form
  $\frac{r[\overline i]}{l[\overline j]}$, where $l\mapsto r$ is a
  rule and $\overline i,\overline j$ are vector-indices, with one of
  the entries in $v_i$, or is identical to the entry in $v_i$. After
  at most $N\cdot(N-1)+1$ steps, there must exist one vector-index
  $j$, such that there exist indices $i\leq i'<j'\leq i+N\cdot(N-1)+1$
  where $v_{i'}[j]$ and $v_{j'}[j]$ are not identical and based on the
  same entry $\ell\in[0,1]$ from $v_i$, i.e. $v_{i'}[j]$ can be
  written as
  $$v_{i'}[j]=\frac{r_k[j]\cdot r_{k-1}[i_{k-1}]\cdot\ldots\cdot r_1[i_1]}{l_k[i_{k-1}]\cdot\ldots\cdot l_2[i_1]\cdot l_1[i_0]}\cdot \ell$$
  and $v_{j'}[j]$ as
  $$v_{j'}[j]=\frac{r_h'[j]\cdot r_{h-1}'[j_{h-1}]\cdot\ldots\cdot r_1'[j_1]}{l_h[j_{h-1}]\cdot\ldots\cdot l_2'[j_1]\cdot l_1'[j_0]}\cdot \ell$$
  where $k$ and $h$ are at most $N\cdot(N-1)$. 
  
  This can be proven as follows: Each vector $v_i$ has at most $N$
  different entries, so there are at most $N$ different entries from
  $v_i$ an entry in a vector $v_j$, $j\geq i$, can be based on. In
  each step, at least one entry of $v_i$ is rewritten to something
  larger. Each of the $N$ entries of $v_i$ can be rewritten and
  increased in the process at most $N-1$ times, without being based on
  the same element from $v_i$ twice (including the initial entry of
  $v_i$), due to the pigeonhole principle.  Again, since in each
  rewriting step at least one entry gets changed and there are only
  $N$ entries, after $N\cdot(N-1)+1$ rewriting steps, there must be at
  least one entry that was based on the same entry from $v_i$ twice
  and increased inbetween.
	
  Thus, we have 
  $$v_{j'}[j]=\frac{r_h'[j]\cdot r_{h-1}'[j_{h-1}]\cdot\ldots\cdot r_1'[j_1]}{l_h[j_{h-1}]\cdot\ldots\cdot l_2'[j_1]\cdot l_1'[j_0]}\cdot\frac{l_k[i_{k-1}]\cdot\ldots\cdot l_2[i_1]\cdot l_1[i_0]}{r_k[j]\cdot r_{k-1}[i_{k-1}]\cdot\ldots\cdot r_1[i_1]}\cdot v_{i'}[j]$$
  which means, $v_{j'}[j]$ can be obtained from $v_{i'}[j]$ via
  multiplication of at most $N\cdot(N-1)+1$ factors of the form
  $\frac{r[\overline i]}{l[\overline j]}$ and at most $N\cdot(N-1)+1$
  factors of the form $\frac{l[\overline j]}{r[\overline i]}$. Also,
  since $v_{j'}[j]>v_{i'}[j]$, this multiplicand is larger than $1$.
  Observe that due to finiteness of $\mathcal R$, there are only
  finitely many products of at most $N\cdot(N-1)+1$ factors of the
  form $\frac{r[\overline i]}{l[\overline j]}$ and at most
  $N\cdot(N-1)+1$ factors of the form
  $\frac{l[\overline j]}{r[\overline i]}$, so there is a least one
  such factor $\delta>1$. Moreover, this construction works for each
  interval of size $N\cdot(N-1)$. Dividing the whole history of rule
  applications into consecutive chunks of size $N\cdot(N-1)+1$ yields
  infinitely many intervals where in each interval at least one index
  increases by at least factor $\delta$. Since the dimension of $v$ is
  finite, there must be at least one index $j$ which has this property
  infinitely often. That means that $v[j]$ is rewritten to something
  larger than $ \delta^n\cdot v[j]$ for each $n\in\mathbb N_0$.
  However, the sequence $\left\langle \delta^n\right\rangle$ is not
  bounded, therefore $\delta^n\cdot v[j]$ is not bounded by $1$, but
  that is a contradiction to the assumption that the rewriting never
  terminates.

  \medskip

Using this result, we can now go on to show that rewriting terminates for the tropical semiring as well.

  We show this by proving that the tropical semiring
  $(\mathbb R_0^+\cup\{\infty\},\min, +\infty, 0)$ and
  $([0,1],\max, \cdot, 0,1)$ are isomorphic with an isomorphism that
  is compatible with the order and the multiplication, which ensures
  that any given transformation system in one of those semirings can
  do a transformation step iff the transformation system obtained by
  applying the isomorphism to each component of every rule as well as
  the vector under consideration can do one.
	
  We use the bijection
  $f: \mathbb R_0^+\cup\{\infty\}\rightarrow[0,1], f(x)=2^{-x}$ where
  we have extended the power to $-\infty$ via the natural definition
  $2^{-\infty}=0$. Obviously, this function is bijective with the
  inverse being $f^{-1}(x) = -\log_2(x)$, where the base-$2$ logarithm
  is extended to $0$ via $\log_2(0)=-\infty$. Hence we only have to
  prove that $f$ respects the order of the lattices and that it is
  compatible with addition and multiplication.
	
  \begin{itemize}
  \item Let $\ell_1,\ell_2\in \mathbb R_0^+\cup\{\infty\}$ be given,
    then
    $$\ell_1\geq \ell_2\Leftrightarrow
    -\ell_1\leq-\ell_2\Leftrightarrow 2^{-\ell_1}\leq
    2^{-\ell_2}\Leftrightarrow f(\ell_1)\leq f(\ell_2).$$
    Note that the order is swapped between the two semirings, so the
    function indeed is an order-isomorphism.
  \item Let $\ell_1,\ell_2\in \mathbb R_0^+\cup\{\infty\}$ be given,
    then
    $$f(\min\{\ell_1,\ell_2\})=2^{-\min\{\ell_1,\ell_2\}} = 
    \max\{2^{-\ell_1},2^{-\ell_2}\} = \max\{f(\ell_1),f(\ell_2)\}$$
  \item Let $\ell_1,\ell_2\in \mathbb R_0^+\cup\{\infty\}$ be given,
    then
    $$f(\ell_1+\ell_2)=2^{-(\ell_1+\ell_2)} = 2^{-\ell_1} \cdot 
    2^{-\ell_2}=f(\ell_1)\cdot f(\ell_2)$$
  \end{itemize}
\end{proof}


\subsection*{Termination for Lattices}

\begin{lem} 
  \label{lemslatticebool}
  Let $\ell_1, \ell_2\in\mathbb L$, where $\mathbb L$ is a lattice,
  then
  \begin{lemlist}
  \item
    $\ell_1\rightarrow_\mathbb L \ell_2\sqsupseteq
    \ell_2$\label{lemlatticebool1}
  \item
    $\ell_1\rightarrow_\mathbb L\ell_2\sqsupseteq\lfloor\neg
    \ell_1\rfloor$\label{lemlatticebool2}
  \item
    $\ell_1\rightarrow_\mathbb L\ell_2=\lfloor \ell_1\rightarrow_\mathbb
    B\ell_2\rfloor$\label{lemlatticebool3}
  \item
    $\lfloor \neg {\ell_1}\rfloor\sqcup_\mathbb B\ell_2\sqsubseteq
    \ell_1\rightarrow_\mathbb L\ell_2\sqsubseteq\neg \ell_1\sqcup_\mathbb B
    \ell_2$\label{lemlatticebool4}
  \item $\ell_1\rightarrow_\mathbb L\ell_2$ can be written as
    $\ell_1^*\sqcup_\mathbb B \ell_2$ for an
    $\lfloor \neg {\ell_1}\rfloor\sqsubseteq \ell_1^*\sqsubseteq\neg {\ell_1}$.
    \label{lemlatticebool5}
  \end{lemlist}
\end{lem}

\begin{proof}~
  \begin{lemlist}
  \item Every lattice is integral, hence
    $\ell_1\cdot\ell_2 \sqsubseteq \ell_2$. Hence $\ell_2$ is an element of
    the set $\{\ell\mid \ell_1\cdot\ell \sqsubseteq \ell_2\}$, the supremum of
    which is $\ell_1\rightarrow_\mathbb L\ell_2$. Hence
    $\ell_1\rightarrow_\mathbb L\ell_2 \ge \ell_2$.

  \item Per definition,
    $\lfloor \neg {\ell_1}\rfloor\sqcap \ell_1\sqsubseteq\neg {\ell_1}\sqcap
    \ell_2=\bot_\mathbb B$
    and if there were an element $\ell\sqsupset\bot_\mathbb L$ such that
    $\lfloor \neg {\ell_1}\rfloor\sqsupseteq \ell$ and $\ell_1\sqsupseteq \ell$,
    then this would be true in $\mathbb B$, too, and therefore such a
    $\ell$ cannot exist.  Thus, in particular,
    $\lfloor \neg {\ell_1}\rfloor\sqcap \ell_1=\bot_\mathbb L\sqsubseteq
    \ell_2$
    and per definition of $\ell_1\rightarrow_\mathbb L\ell_2$, this
    proves that
    $\lfloor \neg {\ell_1}\rfloor\sqsubseteq \ell_1\rightarrow_\mathbb
    L\ell_2$.
  \item We observe that
    $$\lfloor \ell_1\rightarrow_\mathbb
    B\ell_2\rfloor=\bigsqcup_\mathbb L\{\ell\in\mathbb L\mid \ell\sqsubseteq
    \ell_1\rightarrow_\mathbb B\ell_2\}$$
    and
    $$\ell_1\rightarrow_\mathbb L\ell_2=\bigsqcup_\mathbb
    L\{\ell\in\mathbb L\mid \ell_1\sqcap \ell\sqsubseteq \ell_2\}$$
    We will show that both sets are equal:
    \begin{itemize}
    \item $\subseteq$:
      $$\ell\sqsubseteq (\ell_1\rightarrow_\mathbb B\ell_2) \Rightarrow
      \ell_1\sqcap \ell\sqsubseteq \ell_1\sqcap(\ell_1\rightarrow_\mathbb
      B\ell_2)$$ But then:
      $$\ell_1\sqcap \ell\sqsubseteq \ell_1\sqcap(\ell_1\rightarrow_\mathbb
      B\ell_2)\sqsubseteq \ell_2$$
      And thus
      $\ell\in\{\ell'\in\mathbb L\mid \ell_1\sqcap \ell'\sqsubseteq
      \ell_2\}$.
    \item $\supseteq$:
      $$\ell_1\sqcap \ell\sqsubseteq \ell_2\Rightarrow
      \ell\in\{\ell'\in\mathbb L\mid \ell_1\sqcap \ell'\sqsubseteq
      \ell_2\}\subseteq\{\ell'\in\mathbb B\mid \ell_1\sqcap \ell'\sqsubseteq
      \ell_2\}$$
      Hence $\ell$ is smaller or equal than the supremum of this set.
    \end{itemize}
  \item
    $\lfloor \neg {\ell_1}\rfloor\sqcup_\mathbb
    B\ell_2\underset{\str{\cref{lemlatticebool1}},\str{\cref{lemlatticebool2}}}\sqsubseteq
    \ell_1\rightarrow_\mathbb
    L\ell_2\underset{\str{\cref{lemlatticebool3}}}=\lfloor
    \ell_1\rightarrow_\mathbb B\ell_2\rfloor=\lfloor \neg
    \ell_1\sqcup_\mathbb B \ell_2\rfloor\sqsubseteq\neg \ell_1\sqcup_\mathbb
    B \ell_2$
  \item In this proof we are exclusively computing in $\mathbb B$, so
    we do not point out that $\sqcup=\sqcup_\mathbb B$. We will define
    $\ell_1^*:= ((\ell_1\rightarrow_\mathbb L\ell_2)\sqcap\neg
    {\ell_2})\sqcup\lfloor \neg {\ell_1}\rfloor$
    and first prove that
    $\ell_1\rightarrow_\mathbb L\ell_2 = \ell_1^*\sqcup \ell_2$ and
    then prove that
    $\lfloor \neg {\ell_1}\rfloor \sqsubseteq \ell_1^*\sqsubseteq\neg \ell_1$.

    \smallskip

    For the first part of the proof:
    \begin{align*}
      \ell_1\rightarrow_\mathbb L\ell_2
      &\underset{\str{\cref{lemlatticebool4}}}=
        \ell_1\rightarrow_\mathbb L\ell_2\sqcup\lfloor \neg
        {\ell_1}\rfloor\sqcup \ell_2 \\
      &= ((\ell_2\sqcup\neg {\ell_2})\sqcap(\ell_1\rightarrow_\mathbb L \ell_2))\sqcup\lfloor \neg {\ell_1}\rfloor\sqcup \ell_2\\
      =& ((\ell_1\rightarrow_\mathbb L\ell_2)\sqcap\neg
         {\ell_2})\sqcup ((\ell_1\rightarrow_\mathbb L\ell_2)\sqcap
         \ell_2)\sqcup\lfloor \neg {\ell_1}\rfloor\sqcup \ell_2 \\
      & = ((\ell_1\rightarrow_\mathbb L\ell_2)\sqcap\neg
        \ell_2)\sqcup\lfloor \neg {\ell_1}\rfloor\sqcup \ell_2 \\
      &= \ell_1^* \sqcup \ell_2
    \end{align*}
    Now, obviously,
    $\lfloor \neg {\ell_1}\rfloor\sqsubseteq ((\ell_1\rightarrow_\mathbb
    L\ell_2)\sqcap\neg {\ell_2})\sqcup\lfloor \neg {\ell_1}\rfloor =
    \ell_1^*$,
    since $\lfloor \neg {\ell_1}\rfloor$ is part of the supremum. It
    is only left to be shown that
    $((\ell_1\rightarrow_\mathbb L\ell_2)\sqcap\neg
    {\ell_2})\sqcup\lfloor \neg {\ell_1}\rfloor\sqsubseteq\neg {\ell_1}$
    holds as well. First, we observe, that
    $\lfloor \neg {\ell_1}\rfloor\sqsubseteq\neg {\ell_1}$ per definition, so
    it suffices to show that
    $(\ell_1\rightarrow_\mathbb L\ell_2)\sqcap\neg {\ell_2}\sqsubseteq\neg
    \ell_1$.
    Moreover, it is true that
    $\neg {\ell_1}=\bigsqcup\{\ell\in\mathbb B\mid \ell\sqcap
    \ell_1=\bot\}$. Then a simple computation shows:
    $$(\ell_1\rightarrow_\mathbb L\ell_2)\sqcap\neg {\ell_2}\sqcap
    \ell_1 = ((\ell_1\rightarrow_\mathbb L\ell_2)\sqcap
    \ell_1)\sqcap\neg {\ell_2}\sqsubseteq \ell_2\sqcap\neg
    {\ell_2}=\bot.$$Thus,
    $$(\ell_1\rightarrow_\mathbb L\ell_2)\sqcap\neg
    {\ell_2}\in\{\ell\in\mathbb B\mid \ell\sqcap \ell_1=\bot\}$$
    of which $\neg {\ell_1}$ is the supremum.
  \end{lemlist}
\end{proof}

\subsection*{Proof of \cref{thm-termcond}}

\begin{proof}
  \cref{lemlatticebool5} shows that each multiplicand
  $l\rightarrow_\mathbb L v$ can be written as supremum of $v[x]$ for
  an index $x\in X$ and an element $l^*$ from the finite set $L(l,x)$.
  This set is independent of $v$, the element from the set must
  however be chosen according to $v$. Therefore, each element we
  obtain in rewriting is built as infimum and supremum of finitely
  many elements from $\mathbb B$ and -- using conjunctive normal form
  -- we obtain that we can only build finitely many different
  rewriting results from $v$. Therefore, rewriting terminates for
  every vector $v$.
\end{proof}

\subsection{Up-To Techniques for Weighted Automata}
\label{sec:appendix-applications}



\subsection*{Coinduction and Up-to Techniques}

The soundness of the algorithms in \cref{sec:applications} can be
proved in a clear way by exploiting coinduction and up-to
techniques.  In this appendix we shortly recall the essential results
of the theory developed in \cite{PS11}.  We fix the lattice of
relations over $\mathbb S^X$, $Rel_{\mathbb S^X}$, but the results
expressed here hold for arbitrary complete lattices.

Given a monotone map
$b\colon Rel_{\mathbb S^X} \to Rel_{\mathbb S^X}$, the Knaster-Tarski
fixed-point theorem characterises the greatest fixed-point $\nu b$ as
the union of all post-fixed points of $b$:
$$\nu b= \bigcup \{ R \subseteq \mathbb S^X \times \mathbb S^X \mid R \subseteq b(R) \}\text{.}$$
This immediately leads to the \emph{coinduction proof principle}
\begin{equation}
  \label{eq:coinductionproofprinciple}
  \begin{array}{c}
    \exists R, \; S \subseteq R\subseteq b(R)\\
    \hline 
    S \subseteq \nu b
  \end{array}
\end{equation}
which allows to prove $(v_1, v_2)\in \nu b $ by exhibiting a
post-fixed-point $R$ such that $\{(v_1,v_2)\} \subseteq R$. We call
the post-fixed-points of $b$, \emph{$b$-simulations}.  For a monotone
map $f \colon Rel_{\mathbb S^X} \to Rel_{\mathbb S^X}$, a
\emph{$b$-simulation up-to $f$} is a relation $R$ such that
$R \subseteq b(f(R))$. We say that $f$ is \emph{compatible with $b$} if
$f(b(R)) \subseteq b (f(R))$ for all relations $R$. The following result from \cite{PS11} justifies our interest in compatible up-to techniques.

\begin{thm}\label{thm:coinupto}
  If $f$ is $b$-compatible and $R\subseteq b(f(R))$ then
  $R\subseteq\nu b$.
\end{thm}
The above theorem leads to the coinduction up-to principle
\begin{equation}
  \label{eq:coinductionproofprinciple}
  \begin{array}{c}
    \exists R, \; S \subseteq R\subseteq b(f(R))\\
    \hline 
    S \subseteq \nu b
  \end{array}
\end{equation}

Up-to techniques can be combined in a number of interesting ways. For a map $f\colon \mathit{Rel}_{\mathbb S^X}\rightarrow \mathit{Rel}_{\mathbb S^X}$, the $n$-iteration of $f$ is defined as $f^{n+1}=f\circ f^n$ and $f^0=\mathit{id}$, the identity function. The omega iteration is defined as $f^\omega(R)=\bigcup_{i=0}^\infty f^i(R)$. Given two relations $R$ and $S$, we use $R\bullet S$ to denote their relational composition $\{(x,z)\mid\exists y \mid (x,y)\in R\text{\ and\ }(y,z)\in S\}$. For two functions, $f, g: \mathit{Rel}_{\mathbb S^X}\rightarrow \mathit{Rel}_{\mathbb S^X}$, we write $f\bullet g$ for the function mapping a relation $R$ into $f(R)\bullet g(R)$. 

The following result from \cite{PS11} informs us that compatible up-to techniques can be composed resulting in other compatible techniques.

\begin{lem}
  \label{lemmacompositionality}
  The following functions are $b$-compatible:
  \begin{itemize}
  \item $id$: the identity function; 
  \item $f\circ g$: the composition of $b$-compatible
    functions $f$ and $g$;
  \item $\bigcup F$: the pointwise union of an arbitrary family $F$
    of $b$-compatible functions: $\bigcup F (R) = \bigcup_{f\in
      F}f(R)$;
  \item $f^{\omega}$: the (omega) iteration of a $b$-compatible function
    $f$, defined as $f^\omega(R) = \bigcup_{i=0}^\infty f^i(R)$
  \end{itemize}
  Moreover, if $b(R)\bullet b(S)\subseteq b(R\bullet S)$ for all
  relations $R,S$
  \begin{itemize}
  \item $f\bullet g$: the relation composition of $b$-compatible
    functions $f$ and $g$;
  \end{itemize}
  is $b$-compatible.
\end{lem}

With these results it is easy to prove the soundness of the discussed
algorithms.

\subsection*{Language Equivalence for Weighted Automata}

\subsection*{Proof of \cref{prop:coinduction}}
\begin{proof}
  The key step consists in characterising $\sim$ as $\nu \beq$. This
  follows easily from abstract results (see e.g. \cite{BonchiBBRS12,HermidaJ98}), but a
  concrete proof would proceed as follows:
  \begin{enumerate}
  \item prove that $\beq$ is a co-continuous  function, i.e.,
    $\beq(\bigcap_{i=1}^\infty S_i) = \bigcap_{i=1}^\infty \beq(S_i)$
    whenever $S_{i+1}\subseteq S_{i}$:
		
  \item therefore, by the Kleene fixed-point theorem \cite{dp:lattices-order},
    $\nu \beq = \bigcap_n \beq^n(\top) $ where $\beq^0=Id$ and
    $\beq^{n+1}= \beq\circ \beq^n$;
  \item prove, using induction, that for all $n$,
    $\beq^n(\top) = \{(v_1,v_2) \mid \llbracket v_1 \rrbracket(w) =
    \llbracket v_2 \rrbracket(w) \text{ for all }\linebreak
    \text{words } w\in A^*\ \text{ up to length }n-1 \}$:
  \item conclude by 2 and 3 that
    $\nu \beq = \{(v_1,v_2) \mid \llbracket v_1 \rrbracket =
    \llbracket v_2 \rrbracket \text{ for all word } w\in A^* \}$.
  \end{enumerate}
  By coinduction, the first statement follows.

  For the second statement we have to use coinduction up-to and prove
  $\beq$-compatibility of $c$. The latter follows from abstract
  results \cite{BonchiPPR14}. For a concrete
  proof, one has first to show that the following monotone maps are
  $\beq$-compatible.
  \begin{itemize}
  \item the constant reflexive function: $r(R)=\{(x,x)\mid
       x\in S\}$;
  \item the converse function: $s(R)=\{(y,x)\mid (x,y) \in R \}$;
  \item the squaring function: $t(R)=\{(x,z)\mid \exists y,
      (x,y)\in R \text{ and } (y,z)\in R \}$.
  \item the sum function: $+(R)=\{(v_1+v_2,v_1'+v_2')\mid 
      (v_1,v_2)\in R \text{ and } (v_1',v_2')\in R \}$.
  \item the scalar function: $\cdot (R)=\{(v\cdot s, w\cdot s)\mid 
      (v,w)\in R \text{ and } s\in \mathbb S \}$.
 \end{itemize}
 Then, one observe that
 $c = (\mathit{Id}\, \cup r \cup s \cup\, t\cup + \cup \cdot )^\omega$
 and conclude that $c$ is $b$-compatible by
 \cref{lemmacompositionality}.
\end{proof}

\subsection*{Language Inclusion for Weighted Automata}

\begin{lem}
  \label{lemma:greteastbin} 
  $\nu \bin = \ \precsimu$.
\end{lem}

\begin{proof} 
  The proof proceeds as for the first part of \cref{prop:coinduction}
  by using $\sqsubseteq$ in place of $=$ and $\bin$ in place of $\beq$.
\end{proof}

\begin{lem} \label{compatibilitybin}
  If $\sqsubseteq$ is a precongruence, the following monotone maps are
  $\bin$-compatible:
  \begin{itemize}
  \item the constant ord function: $\sqsubseteq(R)=\{(v_1,v_2)\mid 
    v_1 \sqsubseteq v_2 \}$;
  \item the constant inclusion function:
    $\precsimu(R)=\{(v_1,v_2)\mid v_1 \precsimu v_2 \}$;
  \item the squaring function: $t(R)=\{(v_1,v_3)\mid \exists v_2,
      (v_1,v_2)\in R \text{ and } (v_2,v_3)\in R \}$.
  \item the sum function: $+(R)=\{(v_1+v_2,v_1'+v_2')\mid 
      (v_1,v_2)\in R \text{ and } (v_1',v_2')\in R \}$.
  \item the scalar function: $\cdot\,(R)=\{(v\cdot s, w\cdot s)\mid 
      (v,w)\in R \text{ and } s\in \mathbb S \}$.
 \end{itemize}
\end{lem}

\begin{proof}
  Since $\sqsubseteq$ is a precongruence, if $v_1\sqsubseteq v_2$, then
  $o(v_1)\sqsubseteq o(v_2)$ and for all $a\in A$, $t_a(v_1) \sqsubseteq t_a(v_2)$.
  Which means that $\sqsubseteq \subseteq \bin(\sqsubseteq)$, that is, for all
  relations $R$, $\sqsubseteq(\bin (R)) \subseteq \bin (\sqsubseteq (R))$. This
  proves the first statement. The others are similar.
\end{proof}

\begin{lem}
  \label{lemma:c'compatible} 
  $p$ is compatible.
\end{lem}

\begin{proof}
  Observe that by definition
  $p= (\mathit{Id}\, \cup \sqsubseteq \cup\, t \cup + \cup \cdot)^\omega$. The
  statement follows immediately by \cref{lemmacompositionality} and
  \cref{compatibilitybin}.
\end{proof}

\subsection*{Proof of \cref{thm:langeincl}}

\begin{proof}  
  For soundness, observe that the following is an invariant for the
  while loop at step~\texttt{(3)}.
  \begin{equation}
    R\subseteq \bin (p(R)\cup todo)
  \end{equation}
  If \texttt{HKP} returns $true$ then $todo$ is empty and thus
  $R\subseteq \bin (p(R))$, i.e., $R$ is a $\bin$-simulation up-to
  $p$.  By \cref{lemma:c'compatible}, \cref{thm:coinupto} and
  \cref{lemma:greteastbin}, $v_1\precsimu v_2$.
  
  For completeness, we proceed in the same way as in
  \cref{thm:langequicong}.
\end{proof}

\subsection*{Proof of \cref{prop:Preconalgocorrect}}

\begin{proof}
  This proof is very close in structure to the proofs for
  \cref{theorem-correctness} and \cref{thm:RewritingCorrectness}. We
  will use \cref{lem:ChurchRosser} and \cref{lem:prepcorrectness}
  because the claims and proofs for these lemmas can be copied
  verbatim for the asymmetric case, so we do not prove these claims
  again. This does not hold true for \cref{theorem-correctness} and
  \cref{thm:RewritingCorrectness}, though, where symmetry is relevant.
  We will now prove the two claims, adjusted to the non-symmetric
  case.
  \begin{itemize} 
  \item \emph{Whenever there exists a vector $v_2'\geq v_1$ such that
      $v_2$ rewrites to $v_2'$ via $\mathcal R$, i.e.,
      $v_2\leadsto^*_{\mathcal{R}} v_2'$, then $(v_1,v_2)\in p(R)$.}
    This is the analogue to \cref{theorem-correctness}.

    \smallskip
		
    We will show that if $v_2\leadsto^*_{\mathcal R}v_2'$, then
    $(v_2',v_2)\in p(R)$. Furthermore $v_1\le v_2'$ implies
    $v_1\ p(R)\ v'_2$ due to rule~(\textsc{Ord}). Transitivity then
    yields $(v_1,v_2)\in p(R)$.

    Assume $v_2\leadsto v'$ via a rule $l\mapsto r$, then $l=w'$,
    $r=w\sqcup w'$ where $(w,w')\in R$, according to definition of
    $\mathcal R$. Due to closure under linear combinations and
    idempotency of $\sqcup$, we have
    $w\sqcup w' \ p(R)\ w'\sqcup w'=w'$, i.e. $r \ p(R)\ l$.
    Therefore, due to closure under scalar multiplication,
    $r\cdot (l\rightarrow v_2) \ p(R)\ l\cdot (l\rightarrow v_2)$,
    due to closure under linear combination
    $v_2\sqcup r\cdot(l\rightarrow v_2) \ p(R)\ v_2\sqcup
    l\cdot(l\rightarrow v_2)$.
    Keeping in mind the definition of $\rightarrow$, we can observe
    that $l\cdot(l\rightarrow v_2)\leq v_2$, and applying this we
    obtain $v_2\sqcup l\cdot(l\rightarrow v_2)=v_2$, therefore
    $v' \ p(R)\ v_2$. Transitivity yields the claim for $\leadsto^*$.

			
  \item \emph{If $v \ p(R)\ v'$, then $\Downarrow\! v'\geq v$ .} This
    is the analogue to \cref{thm:RewritingCorrectness}. We perform a
    proof by structural induction on the modified derivation rules of
    \cref{tab:proof-rules} (where (\textsc{Sym}) is removed and
    (\textsc{Refl}) is replaced by rule~(\textsc{Ord})).
    \begin{description}
    \item[(\textsc{Rel})] If $v \ p(R)\ v'$ because $v R v'$ then
      there exists a rewriting rule
      $(v'\mapsto v\sqcup v')\in\mathcal R$. Applying this rule shows
      that $\Downarrow\! v'\geq v'\sqcup v\geq v$.
    \item[(\textsc{Ord})] If $v \ p(R)\ v'$ because of the ordering
      rule, i.e.  $v\le v'$, then
      $\Downarrow\!v'\geq v'=v$.
    \item[(\textsc{Trans})] If $v_1 \ p(R)\ v_3$ because of
      $v_1 \ p(R)\ v_2$ and $v_2 \ p(R)\ v_3$, then
      $\Downarrow\!v_3\geq v_2$ implies
      $\Downarrow\!v_3\geq \Downarrow\! v_2$.
      Furthermore $v_1 \ p(R)\ v_2$ inductively yields
      $\Downarrow\!v_2\geq v_1$ and transitivity therefore
      yields $\Downarrow\!v_3\geq v_1$.
    \item[(\textsc{Sca})] If $v\cdot \ell \ p(R)\ v'\cdot \ell$ because
      $v \ p(R)\ p$, then $v\leq\Downarrow\!v'$ and
      \cref{lem:prepcorrectness} yields
      $\Downarrow\!(v'\cdot \ell)\ge (\Downarrow\!
      v')\cdot \ell \geq v\cdot \ell$.
    \item[(\textsc{Plus})] If
      $\overline v\sqcup v \ p(R)\ \overline v'\sqcup v'$ because
      $\overline v \ p(R)\ \overline v'$ and $v \ p(R)\ v'$, then
      $v\sqcup \overline v\leq(\Downarrow_{\mathcal
        R}v')\sqcup(\Downarrow\!\overline
      v')\leq\Downarrow\!(v'\sqcup\overline v')$,
      due to the monotonicity of $\Downarrow$.
    \end{description}
  \end{itemize}
\end{proof}

\subsection*{Threshold Problem for Automata over the Tropical
  Semiring}

\subsection*{Proof of \cref{lems:Atacompat}}

\begin{proof}~
  \begin{lemlist}
  \item Observe that $\mathcal A$ is increasing and thus, if
    $\mathcal A(t_a(v))[x]=\infty$,
    $\mathcal A(t_a(\mathcal A(v)))[x]=\infty$ necessarily holds, too.
    Also note that $\mathcal A(t_a(v))[x]>T$ implies
    $\mathcal A(t_a(v))[x]=\infty$.  Thus we can prove this lemma by
    showing the following: Let $I=\{x\in X\mid u[x]>T\}$, $x\in X$ and
    $a\in A$ such that $t_a(v)[x]\leq T$ be given, then
    $t_a(\mathcal A(v))[x]=t_a(v)[x]$.
    
    All entries greater than $T$ necessarily get mapped to $\infty$.
    
    We first compute:
    $$t_a(\mathcal A(v))[x]=\bigsqcup\{\mathcal A(v)[y]\sqcap
    t_a[y,x]\mid y\in X\}=\min\{\mathcal A(v)[y]+t_a[y,v]\mid
    y\in X\}$$
    and analogously
    $$t_a(v)[x]=\min\{v[y]+t_a[y,x]\mid y\in X\}$$ Since we know
    that $t_a(v)[x]\leq T$, there must exist a $y\in X$ such that
    $v[y]+t_a[y,x]\leq T$. Observe that $v[x]>T$ and
    $t_a[y',x]\in\mathbb N_0$ for all $y'\in I$, therefore $y\notin I$.
    Thus$$t_a(v)[x]=\min\{v[y]+t_a[y,x]\mid y\in X\setminus
    I\}$$ Now we can compute:
    \begin{equation*}
      \begin{aligned}
        &t_a(\mathcal A(v))[x]=\min\{\mathcal A(v)[y]+t_a[y,x]\mid y\in X\}\\
        =&\min\{\min\{\mathcal A(v)[y]+t_a[y,x]\mid y\in X\setminus I\},\min\{\mathcal A(v)[y']+t_a[y',x]\mid y'\in I\}\\
        =&\min\{t_a(v)[x],\min\{\infty+t_a[y',x]\mid y'\in I\}\}=\min\{t_a(v)[x],\infty\}=t_a(v)[x]
      \end{aligned}
    \end{equation*}
  \item First we show that if $o(v)\leq T$ it also holds that
    $o(\mathcal A(v))=o(v)$. If $o(v)\leq T$ then
    $\min\{o[i]+v[i]\mid 1\leq i\leq |X|\}\leq T$, so there is
    an index $i$ where the minimum is reached and smaller than or equal
    $T$, i.e. $o[i]+v[i]\leq T$. Thus, since $+$ is increasing,
    $v[i]\leq T$ and therefore $\mathcal A(v[i])=v[i]$. Since also
    $\mathcal A$ is only increasing, we can conclude
    $o(\mathcal A(v))=o(v)$.
    
    It remains to be shown that
    $o(v)>T\Leftrightarrow o(\mathcal A(v))> T$. Certainly, if $o(v)>T$
    it must follow that $o(\mathcal A(v))> T$, since
    $\mathcal A(v)\geq v$. So assume now, for the converse direction,
    that $ o(\mathcal A(v))> T$ holds. This means
    $\min\{o[i]+\mathcal A(v)[i]\mid 1\leq i\leq |X|\}>T$ and
    assume $o(\mathcal A(v))\neq o(v)$. Then there exists an index $i$
    where the minimum is reached, i.e an index such that
    $o[i]+v[i]\leq o[j]+v[j]$ for all $j$. Since $\mathcal A$ is only
    increasing the entries of a vector, it follows
    $o[i]+v[i]<o[j]+\mathcal A(v[j])$ for all $j$. Thus,
    $o[i]+v[i]<o[i]+\mathcal A(v[i])$, i.e. $v[i]<\mathcal A(v[i])$. It
    follows that $\mathcal A(v[i])=\infty$. Then, $v[i]>T$ and since
    $o[i]\geq0$ it follows that $o(v)>T$.
  \end{lemlist}
\end{proof}

\begin{lem}\label{lemma:c'bullet}
  $p\,\bullet\sqsubseteq$ is $b_2$-compatible, where
  $p\colon Rel_{\mathbb S^X} \to Rel_{\mathbb S^X}$ is the monotone
  function assigning to each relation $R$ its pre-congruence closure.
\end{lem}

\begin{proof}
  By \cref{compatibilitybin} and \cref{lemma:c'compatible},
  $\sqsubseteq$ and $p$ are $b_2$-compatible. It is easy to check that
  for all relations $R, S$, it holds that
  $b_2(R)\bullet b_2(S)\subseteq b_2(R\bullet S)$. Therefore, by
  \cref{lemmacompositionality}, $p\,\bullet\geq$ is $b_2$-compatible.
\end{proof}

\subsection*{Proof of \cref{thm:modifalgcorrect}}
\begin{proof}
Termination is obvious as there are only finitely many vectors with entries from the set $\{0,1,2,.., T, \infty\}$.
  We now prove soundness. Observe that the following is an invariant
  for the while loop at step
  (3).$$R\subseteq b_2((p(R)\cup\mathit{todo})\bullet \sqsubseteq)$$since
  $\mathcal A(t_a(v_1')) \sqsubseteq t_a(v_1') $, i.e.
  $\mathcal A(t_a(v_1')) \geq t_a(v_1')$. If $\HKPA$ returns
  $\mathit{true}$ then $\mathit{todo}$ is empty and thus
  $R\subseteq b_2(p(R)\bullet\sqsubseteq)$. By \cref{lemma:c'bullet},   \cref{thm:coinupto} and
  \cref{lemma:greteastbin},
  $e_t\precsimu v_1$, i.e. $T\geq \llbracket v_1\rrbracket(w)$ for all
  $w\in A^*$.
	
	The converse implication is more elaborated than its analogous in \cref{thm:langequicong}. Assume that $\HKPA$ yields false, then a vector $v_1'$ was found such that $o(v_1')>T$. This means there exists a word $w=a_1a_2...a_n$ such that
\begin{equation*}\begin{aligned}
	v_1'=&t_{a_n}(\mathcal A(t_{a_{n-1}}(\mathcal A(t_{a_{n-2}}(...\mathcal A(t_{a_1}(v_1))...)))))\\
	\leq&\mathcal A(t_{a_n}(\mathcal A(t_{a_{n-1}}(\mathcal A(t_{a_{n-2}}(...\mathcal A(t_{a_1}(v_1))...))))))
\end{aligned}\end{equation*}
Now we can apply \cref{lems:Atacompat} and eliminate all inner $\mathcal A$-applications, yielding $v_1'\leq\mathcal A(t_{a_n}(t_{a_{n-1}}(t_{a_{n-2}}(...(t_{a_1}(v_1))...))))$. For easier reading we will now call \linebreak $v':=t_{a_n}(t_{a_{n-1}}(t_{a_{n-2}}(...(t_{a_1}(v_1))...)))$. Since $o(v_1')>T$, certainly $o(\mathcal A(v'))>T$ due to transitivity. Since $\mathcal A$ is increasing, $o(\mathcal A(v'))\leq\mathcal A(o(\mathcal A(v')))$, which is, according to \cref{lems:Atacompat}, $\mathcal A(o(v'))$. Due to transitivity of $>$ we can conclude $\mathcal A(o(v'))>T$, i.e. $\mathcal A(o(v'))=\infty$. According to the definition of $\mathcal A$, it follows that $o(v')>T$, therefore $w$ is indeed a witness for the automaton not to respect the threshold.
\end{proof}

\subsection*{Exploiting Similarity}


\subsection*{Proof of \cref{simulationimplieslanguageinclusion}}

\begin{proof}
	
  We will prove this inductively, by showing that
  $\llbracket v \rrbracket(w)\sqsubseteq\llbracket v' \rrbracket(w)$ for all
  $w\in\Sigma^*$.
  \begin{itemize}
  \item \textbf{Induction start ($|w|=0$):} In this case,
    $\llbracket v \rrbracket(w)=\llbracket v
    \rrbracket(\epsilon)=o(v)\sqsubseteq o(v')=\llbracket v'
    \rrbracket(\epsilon) = \llbracket v' \rrbracket(w)$.
  \item \textbf{Induction hypothesis:} For all words $w$ where
    $|w|\leq n$ it holds that
    $\llbracket v \rrbracket(w)\sqsubseteq \llbracket v' \rrbracket(w)$.
  \item \textbf{Induction step ($n\rightarrow n+1$):} Let $w$ be given
    such that $|w|=n+1$. Then $w$ can be written as $w=aw'$,
    $a\in\Sigma$, $w'\in\Sigma^*$. Since $(v,v')\in R$, there must
    exist $(u,u')$, $(v_1,v_1'),(v_2, v_2'),\dots,(v_m,v_m')\in R$,
    $s_1, s_2, \dots, s_m\in\mathbb L$ such that
    $u=\bigsqcup\{v_i\cdot s_i\mid 1\leq i\leq m\}$,
    $u'=\bigsqcup\{v_i'\cdot s_i\mid 1\leq i\leq m\}$ and
    $t_a(v)\sqsubseteq u$, $u'\sqsubseteq t_a(v')$. Using monotonicity of $\cdot$
    and $\sqcup$ (wrt.\ $\sqsubseteq$), we obtain the following two
    inequalities:
    $$\llbracket v \rrbracket(w) = \llbracket v \rrbracket(aw') = 
    \llbracket t_a(v) \rrbracket(w')\sqsubseteq \llbracket u \rrbracket(w')$$
    $$\llbracket u' \rrbracket(w')\sqsubseteq \llbracket t_a(v')
    \rrbracket(w') = \llbracket v' \rrbracket(aw') = \llbracket v'
    \rrbracket(w)$$
    Now we can apply the induction hypothesis, keeping in mind that
    $|w'|=n$: Since
    $\llbracket v_i \rrbracket(w') \sqsubseteq \llbracket v'_i
    \rrbracket(w')$
    for all $1\leq i\leq m$. Applying again monotonicity of $\cdot$
    and $\sqcup$, as well as the definition of $u$ and $u'$, we get
    $\llbracket u \rrbracket(w')\sqsubseteq\llbracket u' \rrbracket(w')$.
    Finally, transitivity yields
    $\llbracket v \rrbracket(w)\sqsubseteq\llbracket v' \rrbracket(w)$.
	\end{itemize}
\end{proof}

\subsection*{Proof of \cref{lem:greatestsim}}

\begin{proof}
  First observe that since simulations are closed under union, there
  exists a greatest simulation relation on unit vectors.
  \begin{itemize}
  \item Due to the nature of the algorithm, it necessarily terminates
    after finitely many steps: In the beginning, $R$ contains only
    finitely many pairs of vectors and in each iteration, either some
    pairs are removed from $R$, or $R$ does not change, but in the
    latter case the algorithm terminates.
	
  \item The result is always a simulation relation, this can be seen
    as follows: Let $R$ be the result of a run of the algorithm in
    \cref{fig:sim} and $(v,v')\in R$. Then
		
    \begin{itemize}
    \item It holds that $o(v)\sqsubseteq o(v')$, because otherwise,
      the pair $(v,v')$ would have been removed from $R$ in the first
      foreach loop.
    \item Furthermore 
      $$t_a(v)\sqsubseteq \bigsqcup\{v_1\cdot(v_2\rightarrow
      t_a(v'))\mid (v_1, v_2)\in R\}=:u$$
      Moreover, for all $(v_1, v_2)\in R$,
      $v_2\cdot(v_2\rightarrow t_a(v'))\sqsubseteq t_a(v')$ holds per
      definition of residuation. Therefore,
      $$u':=\bigsqcup\{v_2\cdot(v_2\rightarrow t_a(v'))\mid (v_1, v_2)\in R\}\sqsubseteq t_a(v')$$
      holds, as well. This means that $(u,u')$ is a linear
      combination of $R$-vectors and $t_a(v)\sqsubseteq u$,
      $u'\sqsubseteq t_a(v')$.
    \end{itemize}
  \item Now we will show inductively, that there cannot exist any
    greater simulation relation. The algorithm starts with the full
    cross-product of unit vectors and removes all pairs $(v,v')$ where
    $o(v)\sqsubseteq o(v')$ does not hold -- meaning that these pairs cannot
    be in a simulation relation at all. Thus, before we enter the
    nested foreach loops, $R$ is a superset of (or equal to) the
    greatest simulation relation on $\alpha$. We will now show that,
    whenever a pair of vectors is removed in the nested foreach loops,
    it cannot be contained in a simulation relation, therefore proving
    that after each execution of the nested foreach loops, $R$ retains
    the property to be a superset of (or equal to) the greatest
    simulation relation on $\alpha$.
	
    Assume that $(v,v')$ is an element of the greatest simulation
    relation
    $R'=\{(v_1,v_1'), \linebreak (v_2,v_2'), \dots, (v_n, v_n')\}$.
    Then there must exist multiplicands $s_1, s_2, \dots, s_n$ such
    that $t_a(v)\sqsubseteq \bigsqcup\{v_i\cdot s_i\mid 1\leq i\leq n\}$ and
    $\bigsqcup\{v_i'\cdot s_i\mid 1\leq i\leq n\}\sqsubseteq t_a(v')$. From
    this it follows that for any given $1\leq i\leq n$, it holds that
    $v_i'\cdot s_i\sqsubseteq t_a(v')$ and therefore
    $s_i\sqsubseteq \bigsqcup\{\ell\mid v_i'\cdot\ell\sqsubseteq
    t_a(v')\}=v_i'\rightarrow t_a(v')$,
    since it is included in the set. We can therefore compute, using
    $R'\sqsubseteq R$, which is the induction hypothesis:
    \begin{equation*}
      \begin{aligned}
        u &= \bigsqcup\{v_1\cdot(v_2\rightarrow t_a(v'))\mid (v_1,
        v_2)\in R\}\\&\geq\bigsqcup\{v_1\cdot(v_2\rightarrow
        t_a(v'))\mid (v_1, v_2)\in
        R'\}\\&=\bigsqcup\{v_i\cdot(v_i'\rightarrow t_a(v'))\mid 1\leq
        i\leq n\}\\&\geq\bigsqcup\{v_i\cdot s_i\mid 1\leq i\leq
        n\}\geq t_a(v)
      \end{aligned}
    \end{equation*} Thus, the
    if-condition $t_a(v)\not\sqsubseteq u$ in the second-to-last line does not
    evaluate to true, meaning that $(v,v')$ will not be removed from $R$
    in the current iteration.
	\end{itemize}
\end{proof}

\subsection*{Proof of \cref{lem:sim-complexity}}

\begin{proof}
  Assuming all semiring operations (supremum, multiplication,
  residuation) consume constant time, the runtime of the algorithm
  can be analysed as follows: the for-loop in line~\texttt{(3)} is
  executed $|X|^2$ many times, once for each element of $R$, which is
  initialised as all pairs of unit vectors of dimension $|X|$, of
  which there are exactly $|X|$ many. The while-loop in
  line~\texttt{(4)} is executed until $R$ remains constant for one
  iteration. Since $R$ contains at most $|X|^2$ many elements in the
  beginning (if no pair was thrown out in the preceding for-loop), and
  in each step of the inner for all-loop, \texttt{(4.2.1)}, $R$
  remains either constant or a pair of vectors gets taken out of $R$.
  The for all-loop in line 4.2.1 is executed $|A|$-times each time,
  and the inner for all loop is executed $|R|$-times. While line
  \texttt{(4.2.1.2)} only takes constant time, line \texttt{(4.2.1.1)}
  takes $|X|\cdot|R|$ many steps. At worst, the time taken by the
  whole while-loop in line~\texttt{(4)} therefore is
  $|A|\cdot\sum_{i=1}^{|X|^2}((|X|^2-i)^2\cdot|X|)\in\mathcal
  O(|A|\cdot|X|^7)$,
  if in each iteration exactly one pair of vectors gets removed from
  $R$. Obviously, the latter loop dominates the former, so the run
  time is in $\mathcal O(|A|\cdot|X|^7)$.
\end{proof}

\subsection*{Proof of \cref{lemmaHKP'}}

\begin{proof}
  For soundness, we use the invariant
  $R\subseteq \bin (p'(R)\cup todo)$, which allows to conclude that
  $R\subseteq \bin (p'(R))$. Now $p'$ is not guaranteed to be
  compatible, but $p''$ defined for all relations $R$ as
  $p''(R)=p(R\, \cup \precsimu)$ is compatible by
  \cref{compatibilitybin} and \cref{lemmacompositionality}. By
  \cref{simulationimplieslanguageinclusion},
  $\preceq \subseteq \precsimu$. By monotonicity of $p$, we have that
  $p'(R)=p(R\cup \preceq) \subseteq p(R\cup \precsimu) = p''(R)$.
  By monotonicity of $\bin$, $R\subseteq \bin (p''(R))$. By
  \cref{thm:coinupto} and \cref{lemma:greteastbin}, $v_1\precsimu v_2$.
  For completeness, we proceed in the same way as in
  \cref{thm:langequicong}.
\end{proof}

\subsection*{Proof of \cref{thm:HKPA'sound}}

\begin{proof}
  For termination and completeness we reuse the same argument as in
  \cref{thm:modifalgcorrect}.  For soundness we need to combine the
  proof of \cref{lemmaHKP'} and of \cref{thm:modifalgcorrect}: we use
  the invariant
  $R\subseteq \bin ((p'(R)\cup todo)\bullet \sqsubseteq)$, which
  allows to conclude that
  $R\subseteq \bin (p'(R)\bullet \sqsubseteq)$. Now $p'$ is not
  guaranteed to be $b_2$-compatible, but $p''$, as defined in the
  proof of \cref{lemmaHKP'} is.  By \cref{lemmacompositionality},
  $p'' \bullet \sqsubseteq$ is compatible, since $ \sqsubseteq$ is
  compatible.  By monotonicity of $p$, we have that
  $p'(R) \bullet \sqsubseteq\ =\ p(R\ \cup \preceq) \bullet
  \sqsubseteq \ \subseteq\ p(R\ \cup \precsimu)\bullet \sqsubseteq\ =\
  p''(R)\bullet \sqsubseteq$.
  By monotonicity of $\bin$,
  $R\subseteq \bin (p''(R)\bullet \sqsubseteq)$. By
  \cref{thm:coinupto} and \cref{lemma:greteastbin}, $v_1\precsimu v_2$.
\end{proof}

\section{Shortest Path Problem in Directed Weighted Graphs}
\label{sec:dijkstra}

The rewriting algorithm to compute the congruence closure over the tropical semiring
is closely related to the Dijkstra algorithm to find the shortest paths in directed weighted
graphs.

A \emph{weighted directed graph} is a couple $G=(V, \mathsf{weight})$ where $V=\{ 1, 2, \ldots, n\}$ is the set of vertices and 
 $\mathsf{weight}:E\rightarrow\mathbb R_0^+ \cup \{\infty\}$ is a function assigning to each pair of vertexes $v,w$ the weight to move from $v$ to $w$. Intuitively, the weight is $\infty$ if there is no way (no edges) to go from $v$ to $w$.

\begin{defi}[Rewriting System of a Graph]
  Let $G=(V,E, \mathsf{weight})$ be a weighted, directed
  graph and $e_i$ be the $i$-th unit vector of dimension $|V|$ in
  $\mathbb T$ and fix
  $v_i $ to be the vector representing the outgoing arrows
  from the $i$-th vertex, i.e.
	$v_i[j]=\mathsf{weight}((i, j))$. The rewriting system associated to $G$ is $\mathcal{R}_G=\{e_i \mapsto v_i \mid 1\leq i \leq |V|\}$.
\end{defi}

\begin{prop} 
  For a graph $G$ and vertex $i$, let $\Downarrow\! e_i$ be the normal
  form reached with $\mathcal{R}_G$. Then for all vertexes $j$,
  $\Downarrow\! e_i[j]$ is weight of the shortest path from $i$ to
  $j$.
\end{prop}
%

In order for rewriting to behave exactly like Dijkstra's algorithm, we
need to always choose the rule with the smallest multiplicand that is
applicable -- i.e. the greatest one according to the order of the
lattice. Choosing a rule corresponds to choosing a vertex to
explore, determining the multiplicand corresponds to finding the
weight accumulated from the starting vertex to the vertex to be
explored next and applying the rule corresponds to updating the
distances of all adjacent vertices that can be reached via a shorter
path than the currently known shortest path.
%
%
The following example shows the rewriting at work.

\begin{ex}
	Consider the following directed graph:
	
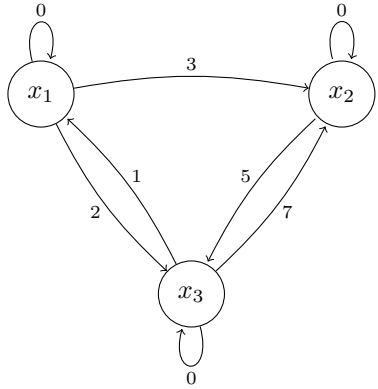
\begin{center}
    \begin{tikzpicture}[x=2cm,y=-2cm,double distance=2pt]
      \node[state] (b) at (1,0) {$x_1$} ;
      \node[state] (d) at (3,0) {$x_2$} ;
      \node[state] (e) at (2,1.33) {$x_3$} ;
      \begin{scope}[->]
        \path 
          (b) edge[bend right=10] node[arlab,below left] {$2$} (e) 
              edge[bend left=10] node[arlab,above] {$3$} (d) 
              edge[loop above] node[arlab] {$0$} (b) ;
        \path[shorten <=1pt] 
          (d) edge[bend right=10,shorten >=1pt] node[arlab,above left] {$5$} (e) 
              edge[loop above] node[arlab] {$0$} (d) ;
        \path
          (e) edge[bend right=10,shorten >=1pt] node[arlab,below right] {$7$} (d) 
              edge[bend right=10,shorten >=1pt] node[arlab,above right] {$1$} (b) 
              edge[loop below] node[arlab] {$0$} (e) ;
      \end{scope}
    \end{tikzpicture}
  \end{center}
  This graph corresponds to the following rule system:
  $$\mathcal R=\left\{ \begin{pmatrix}0\\\infty\\\infty
    \end{pmatrix}\mapsto \begin{pmatrix}0\\3\\2
    \end{pmatrix},\begin{pmatrix}\infty\\0\\\infty \end{pmatrix}
    \mapsto \begin{pmatrix}\infty\\0\\5
    \end{pmatrix},\begin{pmatrix}\infty\\\infty\\0
    \end{pmatrix}\mapsto \begin{pmatrix}1\\7\\0 \end{pmatrix}
  \right\}$$
  Now we can, for instance, determine the weights of the shortest paths to all vertices starting in the vertex $x_3$ as follows:
  $$\begin{pmatrix}\infty\\\infty\\0 \end{pmatrix}\leadsto
  1+\begin{pmatrix}0\\\infty\\\infty
  \end{pmatrix}\min\begin{pmatrix}1\\7\\0\end{pmatrix} = \begin{pmatrix}1\\7\\0\end{pmatrix} \leadsto 1+\begin{pmatrix}0\\3\\2 \end{pmatrix}\min\begin{pmatrix}1\\7\\0\end{pmatrix}=\begin{pmatrix}1\\4\\0\end{pmatrix}$$
  Afterwards, no more rewriting rule is applicable.

  Note that rewriting is non-deterministic and we could have chosen
  another path.  If we choose to apply the rule that has smallest
  coefficient, we effectively simulate Dijkstra's algorithm.

\end{ex}

}

\end{document}